\documentclass[a4paper, 12pt]{article}

\usepackage{graphicx}
\usepackage{xcolor}
\usepackage{cancel}
\usepackage{graphicx}
\usepackage{subcaption}
\usepackage{ifthen}
\usepackage{amssymb}
\usepackage{amsmath}
\usepackage{amsthm}
\usepackage{enumitem}
\usepackage[utf8]{inputenc}
\usepackage[titletoc]{appendix}
\usepackage[ruled, linesnumbered,vlined]{algorithm2e}

\newtheorem{theorem}{Theorem}
\newtheorem{definition}{Definition}
\newtheorem{lemma}{Lemma}
\newtheorem{corollary}{Corollary}

\usepackage{lineno,hyperref}
\hypersetup{pdfauthor={Name}}
\modulolinenumbers[5]

\usepackage{authblk}

\begin{document}

\title{Computing the Best Case Energy Complexity of Satisfying Assignments in Monotone Circuits}

\author[1]{Janio Carlos Nascimento Silva}
\author[2,5]{Uéverton S.  Souza}
\affil[1]{Instituto Federal do Tocantins, Campus Porto Nacional, Porto Nacional, Brazil}
\affil[2]{Institute of Informatics, Faculty of Mathematics, Informatics and Mechanics, University of Warsaw, Warsaw, Poland}
\affil[5]{Instituto de Computação, Universidade Federal Fluminense, Niterói, Brasil}

\begin{titlepage}

\maketitle

\begin{abstract}
Measures of circuit complexity are usually ana\-lyzed to ensure the computation of Boolean functions with economy and efficiency. One of these measures is the {\em energy complexity}, which is related to the number of gates that output {\tt true} in a circuit for an assignment.
The idea behind {energy complexity} comes from the counting of `firing' neurons in a natural neural network. The initial model is based on threshold circuits, but recent works also have analyzed the energy complexity of traditional Boolean circuits.
In this work, we discuss the time complexity needed to compute the best case energy complexity among satisfying assignments of a monotone Boolean circuit, and we call such a problem as {\sc MinEC}$^+_M$.
In the {\sc MinEC}$^+_M$ problem, we are given a monotone Boolean circuit $C$, a positive integer $k$ and asked to determine whether 
there is a satisfying assignment $X$ for $C$ such that $EC(C,X) \leq k$, where $EC(C,X)$ is the number of gates that output {\tt true} in $C$ according to the assignment $X$.
We prove that {\sc MinEC}$^+_M$ is {NP-complete} even when the input monotone circuit is planar. Besides, we show that the problem is W[1]-hard but in XP when parameterized by the size of the solution. In contrast, we show that when the size of the solution and the genus of the input circuit are aggregated parameters, the {\sc MinEC}$^+_M$ problem becomes fixed-parameter tractable.
\end{abstract}

\noindent {\bf Keywords:}\ energy complexity, monotone circuit, planar, genus, FPT.

\end{titlepage}


\section{Introduction}
\label{sec:intro}

Circuit Complexity is a research field that aims to study bounds for measures (such as size and depth) of circuits that compute Boolean functions. The \emph{size} of a circuit is its number of logic gates, and its \emph{depth} is the largest path from any input to the output gate. A circuit complexity analysis can provide lower/upper bounds on circuits classes that represent classic decision problems besides the possibility to identify efficient Boolean circuits according to specific properties (see~\cite{hastad1986almost,razborov1987lower,williams2014nonuniform}). 
In addition, some important bounds are described to deal with different definitions of size~\cite{allender2010amplifying}. 
From a combinatorial point of view, several optimization problems address the minimization of measures in some circuits classes, such as the notorious \textsc{Weighted Circuit Satisfiability} problem, where the weight measures the amount of {\tt true} values assigned to the input variables. 

Despite this `zoo of measures', optimizing properties like size and depth does not always guarantee an ‘efficient' design of a specific circuit class. Depending on the purpose, a circuit with small size (either considering gates or wires) or depth can be inappropriate; such a situation was identified in~\cite{uchizawa2006computational}.
When faced with threshold circuits used as an artificial neural network, it is possible to observe a contrast with neurons of the human brain. The authors in~\cite{uchizawa2006computational} (based in neuroscience literature) argue that the activation of neurons in a human brain happens sparsely. It was shown in \cite{lennie2003cost} that the metabolic cost of a single spike in cortical computation is very high in a way that approximately $1\%$ of the neurons can be activated simultaneously.
This phenomenon happens due to the asymme\-tric energy cost between neurons activated and non-activated in natural cases.
From the other side, digital circuits, when satisfied (outputting {\tt true}), on average activate $50\%$ of the gates.  

Under different perspectives, `energy' (or `power') of a circuit is a measure that has a lot of attention in the literature. Due to multiple models (from biology, electronics, or purely theoretical), several works address different ways of analyzing the energy of a circuit. In \cite{kissin1982measuring}, the energy consumption of a circuit considers the switching energy consumed by wires (edges) and gates of VLSI circuits. In \cite{antoniadis2014energy} and \cite{barcelo2015almost}, it is analyzed the \textit{voltage-to-energy} consumed by the gates, taking into consideration the \textit{failure-to-energy}. Other different models are explored in~\cite{vsima2014energy} and \cite{blake2017vlsi}, such works try to explore concepts of energy too intrinsic to the design of practical circuits on electronics. 

In this paper, we deal with a circuit complexity measure called \emph{energy complexity} ($EC$). The idea behind this measure is to evaluate the number of gates in a circuit that returns {\tt true} for an assignment. A similar concept called `power of circuit' was studied by \cite{vaintsvaig1961power}. The term \emph{energy complexity} was introduced in \cite{uchizawa2006computational} as an alternative to the dilemma \textit{artificial vs natural} described above. In \cite{uchizawa2006computational}, the authors prove that the minimization of circuit energy complexity obtains a different structure from the minimization of classical circuit complexity measures and potentially closer to the structure of neural networks in the human brain. The authors proved initial lower bounds for energy complexity and other circuit complexity measure called {\em entropy}. For an analysis more focused on circuit complexity and recent bounds for energy complexity of Boolean functions we refer to~\cite{dinesh2020new}.

With a different perspective, this work dedicates attention to optimization and decision problems related to computing the energy complexity of a circuit.
More precisely, we consider the problem of determining the satisfying assignment with minimum energy consumption in monotone circuits, i.e., the best case energy complexity of a satisfying assignment in the class of monotone circuits, called {\sc MinEC}$^{+}_M$. Our focus is on time (parameterized) complexity of the {\sc MinEC}$^{+}_M$ problem.

The paper's sequence has some preliminaries and the problem definition, a discussion of its NP-completeness on planar circuits, a W[1]-hardness proof and an XP algorithm for the case where the number of gates to be activated is considered as parameter. Finally, we present an FPT algorithm for {\sc MinEC}$^{+}_M$ when parameterized either by the treewidth of the input circuit or by the genus of the input plus the size of the solution.

An extended abstract of this paper has been presented in AAIM 2021 \cite{silva2021energy}.

\subsection*{Preliminaries}

A \emph{Boolean circuit} is a model that computes a Boolean function $g: \{0,1\}^n \rightarrow \{0,1\}$ over a basis of operators (e.g. $\{{\tt AND,OR, NOT},...\}$). In terms of Graph Theory, a Boolean circuit is a directed acyclic graph $C(V,E)$ where the set of vertices $V = (I,G,\{v_{out}\})$ is partitioned into: 
(i) a set of inputs $I = \{i_1,i_2, \dots\}$ composed by the vertices of in-degree $0$;
(ii) a set of gates $G = \{g_1, g_2, \dots\}$,  which are vertices labeled with Boolean operators; (iii) and the single output (sink) vertex $v_{out}$ with out-degree equal to $0$ and also labeled with a Boolean operator (see Fig.~\ref{fig:1}).
The input vertices represent Boolean variables that can take values from $\{0,1\}$ ($\{ ${\tt true}$,$ {\tt false}$\}$, depending on the conventions), and the label/operator of a gate or output vertex $w$ is given by $f(w)$. A \emph{monotone circuit} is a Boolean circuit where the Boolean operators allowed are in $\{{\tt AND,OR}\}$.
An \emph{assignment} of C is a vector $X = [x_1, x_2, \dots, x_{|I|}]$ of values for the set of inputs $I$, where for each $j$, $x_j \in \{0,1\}$ is the value assigned to input $i_j$. We say that $X$ is a satisfying assignment if the circuit $C$ evaluates to $1$ ({\tt true}) when $X$ is given as input.

\begin{figure}[!ht]
\begin{subfigure}{.48\linewidth}\centering \includegraphics[width=0.8\linewidth]{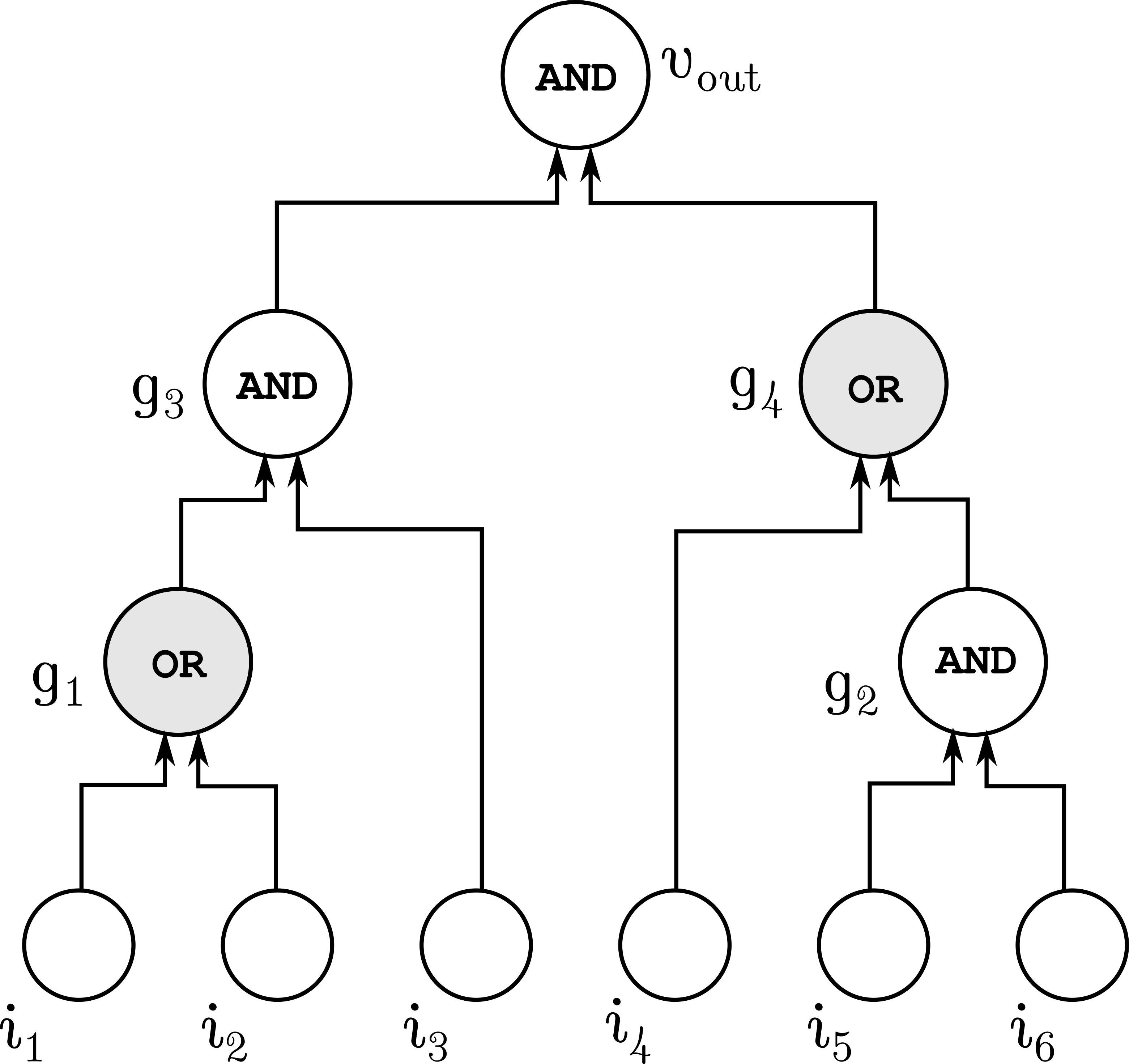} \caption{Graph representation of circuit $C$.}\label{fig:graphC}\end{subfigure}
\hfill
\begin{subfigure}{.48\linewidth}\centering \includegraphics[width=0.8\linewidth]{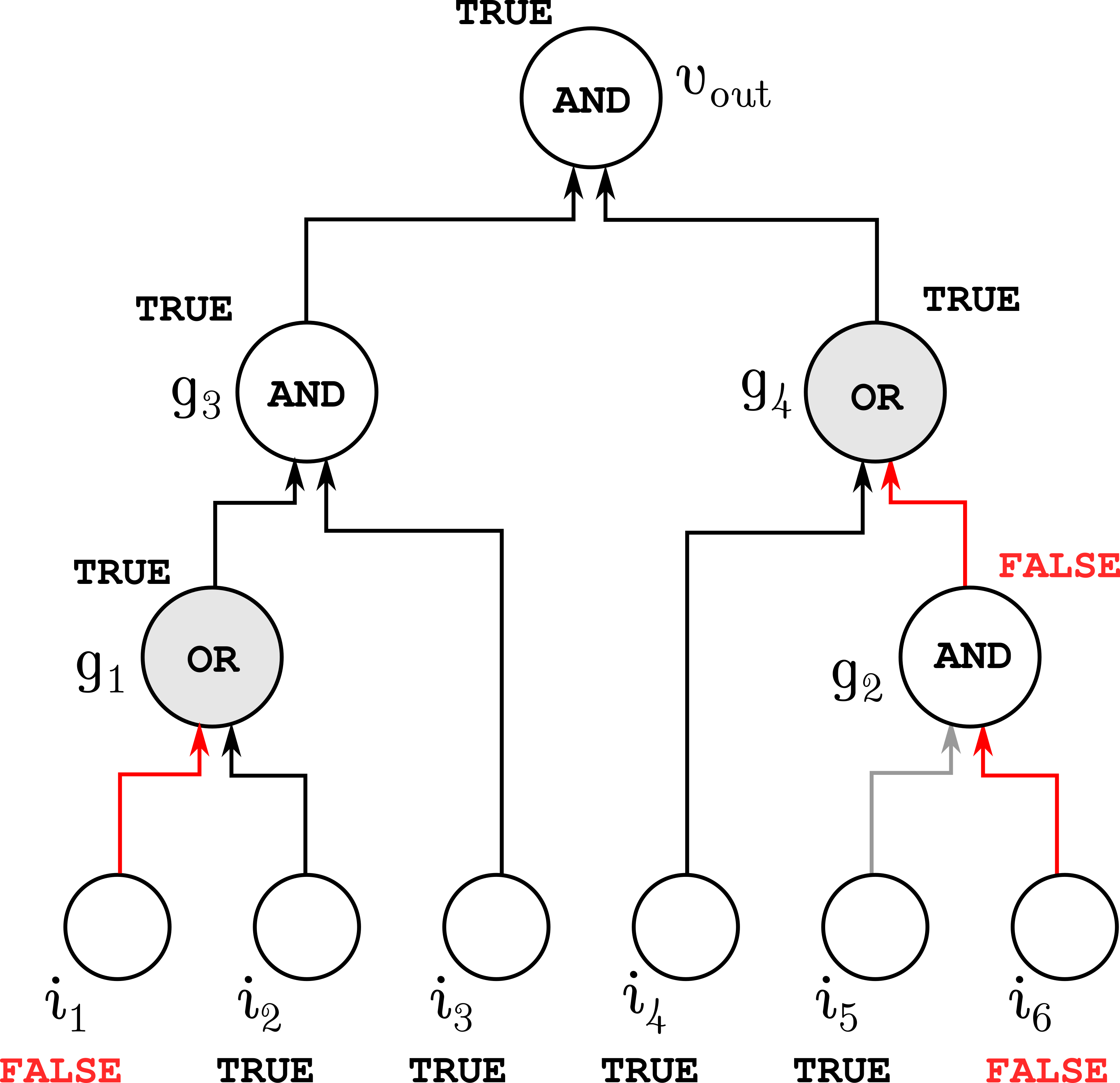} \caption{$C$ under a satisfying assignment.}\label{fig:positiveC}\end{subfigure}
\caption{Graph representation and a satisfying
assignment.}\label{fig:1}
\end{figure}

Given a Boolean circuit $C$ and an assignment $X$, the Energy Complexity of $X$ into $C$, $EC(C,X)$, is defined as the number of gates that output {\tt true} in $C$ according to the assignment $X$. The (Worst-Case) Energy Complexity of $C$ (denoted by $EC(C)$) is the maximum $EC(C, X)$ among all possible assignments $X$ (See \cite{dinesh2020new}). Analogously, the Best-Case Energy Complexity of $C$ (denoted by $MinEC(C)$) is the minimum $EC(C, X)$ among all possible assignments $X$. 

In~\cite{ALVES20211,alves2020succinct}, a measure called {\em certification-width} was described, which is the size of a minimum subset of edges that are enough to certificate a satisfying assignment. Such edges form a structure called {\em succinct certificate} that can be seen as a minimal map of edges to be followed in order to activate the output gate. Similar structures, denoted by
accepting subtree, positive proof, or solution subgraph, can also be found in~\cite{venkateswaran1992circuit,dorzweiler2016positive,and_or_graphs2,and_or_graphs1,lima2018and}.

Note that there are similarities between certification-width and energy complexity. Both measures indicate saturation levels of circuits, but while certifica\-tion-width focuses on edges, energy complexity is about the activation of gates. However, energy complexity presents two additional challenges: (i) $EC$ ignores the `fi\-ring' of input gates; (ii) $EC$ counts activated gates even if its signal does not reach the output gate (due to unsatisfied gates -- see Fig.~\ref{fig:anomaly}
). These two issues forbid rushed conclusions about $EC$ based on what we know about certification-width. Nevertheless, the study in~\cite{ALVES20211} also motivates the study of the complexity of computing the best case energy complexity of satisfying assignments in monotone circuits.

\begin{figure}[!ht]
\centering
\includegraphics[width=0.5\textwidth]{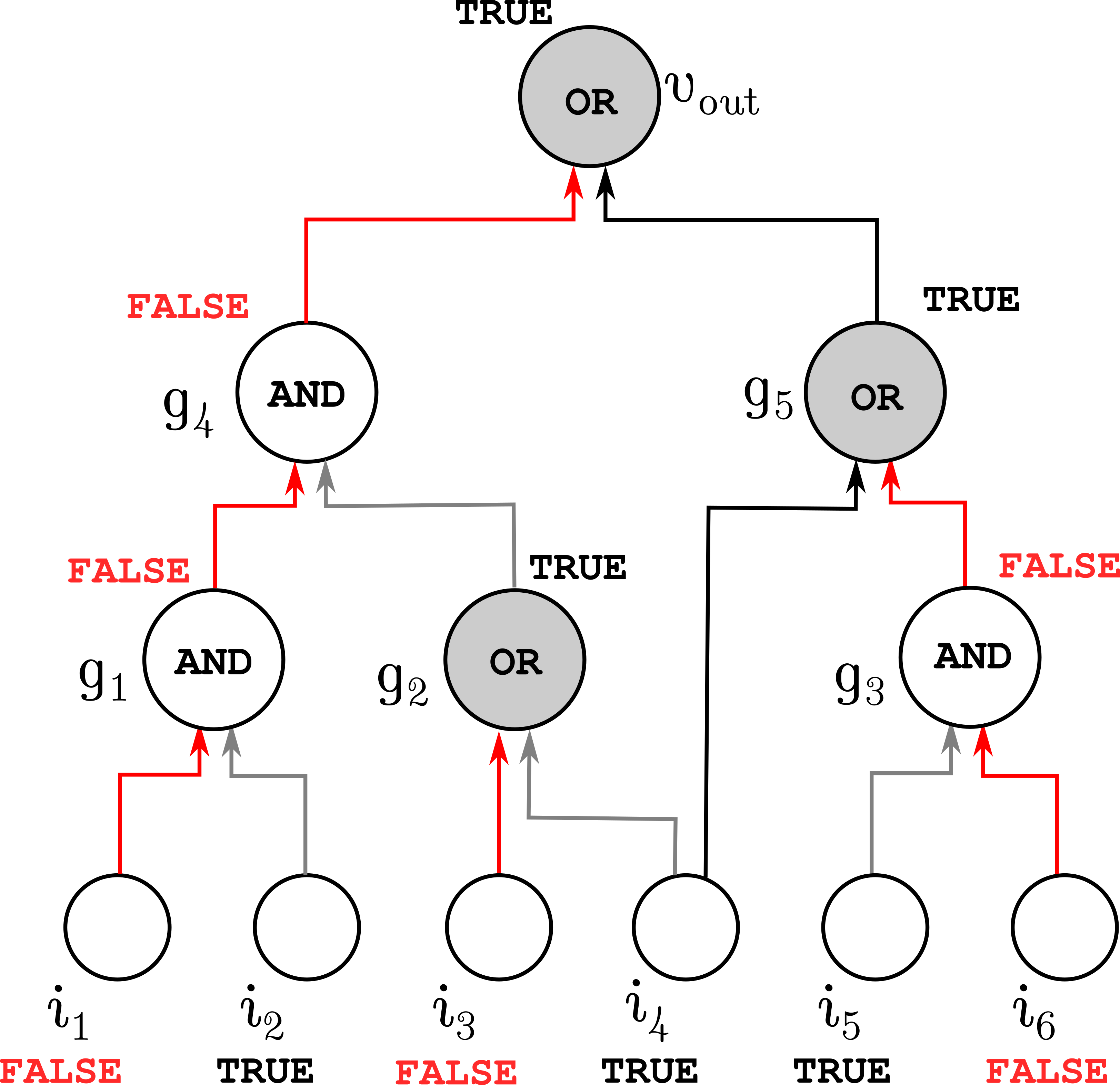}
\caption{Anomalous behavior in certificates for {\sc MinEC}$^+_M$. Note that the edge $(i_4,g_2)$ produces a `leak' of the assignment, i.e. $g_2$ output {\tt true} (increasing the energy complexity), but its signal is not relevant to satisfy $v_{out}$.}
\label{fig:anomaly}
\end{figure}

%

Note that in energy complexity problems in addition to working with the gates needed to satisfy the circuit, it is still necessary to handle gates that assignments may collaterally activate. 
Such behavior makes working with energy complexity problems more challenging than typical satisfying problems where the focus is only on the minimal set of inputs, gates, or wires/edges sufficient to satisfy the circuit.

While computing the worst-case energy complexity of satisfying assignments in monotone circuits is trivial (just activate all inputs), the problem of computing the best-case energy complexity among all satisfying assignments in monotone circuits seems a challenge. Therefore, in this work, we address this particular case where the circuit is monotone, focusing on the following decision problem:


\medskip
\noindent	\fbox{
		\parbox{0.96\textwidth}{
\noindent
{\sc \textsc{Best-Case Energy Complexity of Satisfying Assignments in Monotone Circuits -- MinEC$^+_M$}} 

\noindent
\textbf{Instance}: A monotone Boolean circuit $C$ and a positive integer $k$.

\noindent
\textbf{Question}: Is there a satisfying assignment $X$ for $C$ such that $EC(C,X) \leq k$?
}
}

Besides, we denote by {\sc $k$-MinEC$^+_M$} the parameterized version of {\sc MinEC$^+_M$} where $k$ is taking as the parameter.

\section{Computational complexity analysis}

In this section, we present our (parameterized) complexity results regarding {\sc MinEC$^+_M$}. Since planarity is a notion with significant relevance in the field of circuit analysis, we started our analysis with a focus on planar circuits.

\subsection{NP-hardness on planar circuits}
Using a reduction from \textsc{Planar Vertex Cover}, similar to that employed in \cite{ALVES20211}, we are able to show that {\sc MinEC$^+_M$} is NP-complete even when restrict to planar circuits.
%

\begin{theorem}\label{thm:np}
\textsc{MinEC}$^+_M$ is NP-complete even when restricted to planar circuits.
\end{theorem}

\begin{proof}
Given a circuit $C$ and an integer $k$, forming an instance of \textsc{MinEC}$^+_M$, an assignment of Boolean values can be seen as a certificate for this instance. Since it is easy to count the number of gates outputting \texttt{true} according to an assignment, clearly, \textsc{MinEC}$^+_M$ is in NP.


Now, we show the NP-hardness of \textsc{MinEC}$^+_M$ by a reduction from \textsc{Planar Vertex Cover (PVC)}, a well-known NP-complete problem.

\medskip
\noindent	\fbox{
		\parbox{0.96\textwidth}{
\noindent
{\sc \textsc{Planar Vertex Cover (PVC)}}

\noindent
\textbf{Instance}: A planar graph $G$; a positive integer $c$.

\noindent
\textbf{Question}: Is there a set $S$ of size at most $c$, such that for each edge $(u,v) \in E(G)$ either $u \in S$ or $v\in S$?
}
}
\medskip

First, consider the following preprocessing: let $(H,c')$ be an instance of {\sc Planar Vertex Cover}, by subdividing twice each edge of $H$, we obtain a graph $G$ where each edge $e=(ab)$ of $H$ is replaced by a $P_4$ $ab'a'b$, where $a'$ and $b'$ are new vertices (see Fig.~\ref{fig:instancePVC}). 

\begin{figure}[!ht]
\begin{subfigure}[b]{.48\linewidth}\centering \includegraphics[width=0.8\linewidth]{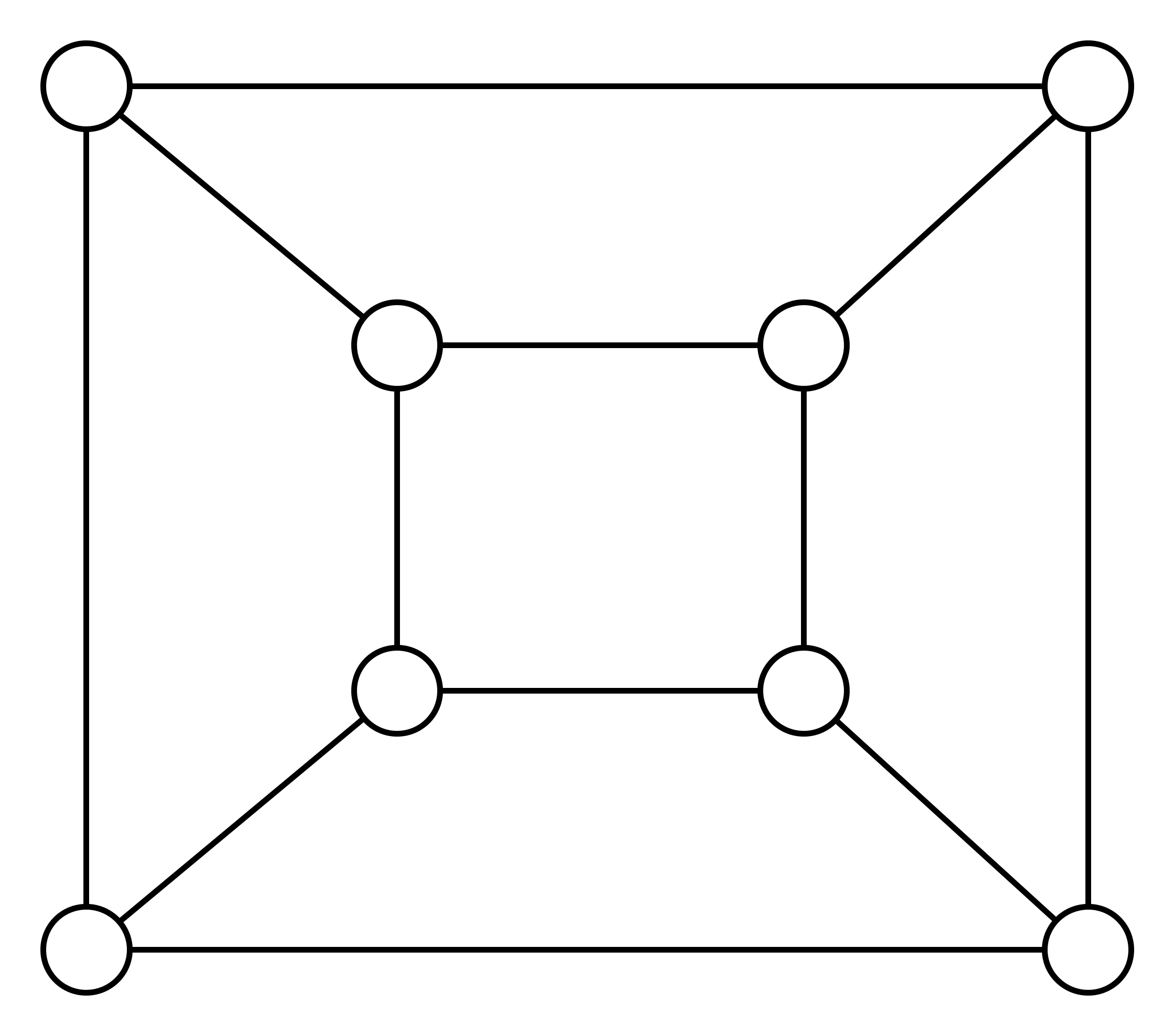} \caption{Graph $H$.}\label{original-H}\end{subfigure}
\hfill
\begin{subfigure}[b]{.48\linewidth}\centering \includegraphics[width=0.8\linewidth]{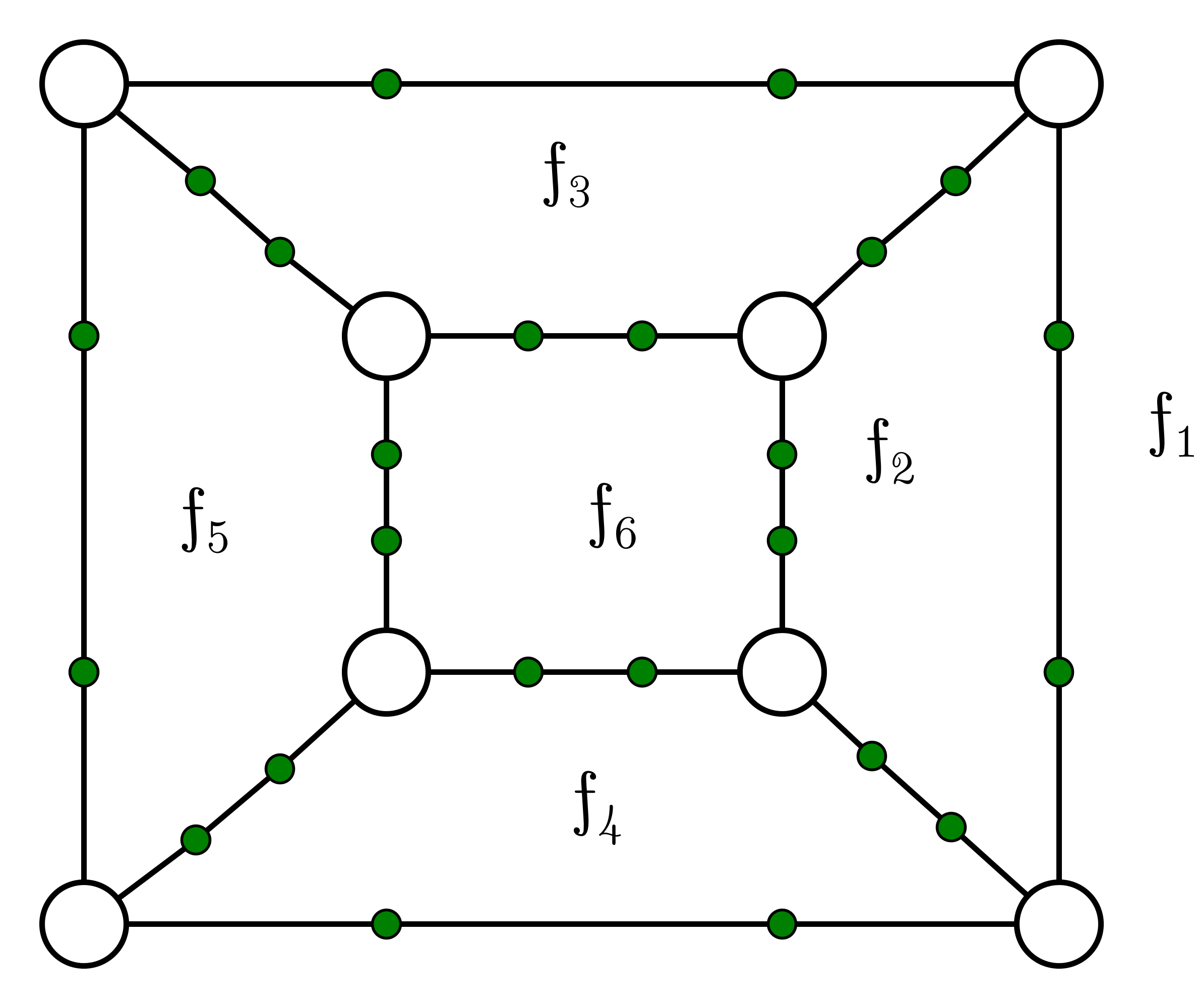} \caption{Graph $G$.}\label{graph-G}\end{subfigure}
\caption{Green vertices are those $a'$ and $b'$ inserted in $G$ after the preprocessing step.}
\label{fig:instancePVC}
\end{figure}

Notice that $G$ is also planar; $H$ has a vertex cover of size $c'$ if and only if $G$ has a vertex cover of size $c=c'+|E(H)|$; and given a planar embedding of $G$, the boundary of any pair of adjacent faces of $G$ contains at least three edges.

From a fixed planar embedding of the instance $(G,c)$ of \textsc{Planar Vertex Cover}, we proceed with the reduction.  We will construct an instance $(C, k)$ of \textsc{MinEC}$^+_M$ where $C$ is a planar monotone circuit, and $k$ is the target size of the energy complexity.

From the structure of $G$, we apply the following:

\begin{enumerate}
\item first, set $V(C)=V(G)$;
\item for each vertex $v_i \in V(G)$, create an input vertex $v_{i}^{in}$, assign $f(v_{i}) = \texttt{AND}$, and add a directed edge $(v_{i}^{in},v_{i})$; 

\item for each edge $e_i=(u,v) \in E(G)$, create a vertex $v_{e_i}^{cover}$ such that $f(v_{e_i}^{cover}) = \texttt{OR}$, and create the directed edges $(v, v_{e_i}^{cover})$ and $(u, v_{e_i}^{cover})$. (see Fig.~\ref{step3})

\begin{figure}[!ht]\label{fig:step3ETDG}
\begin{subfigure}[b]{0.48\linewidth}\centering \includegraphics[width=0.8\linewidth]{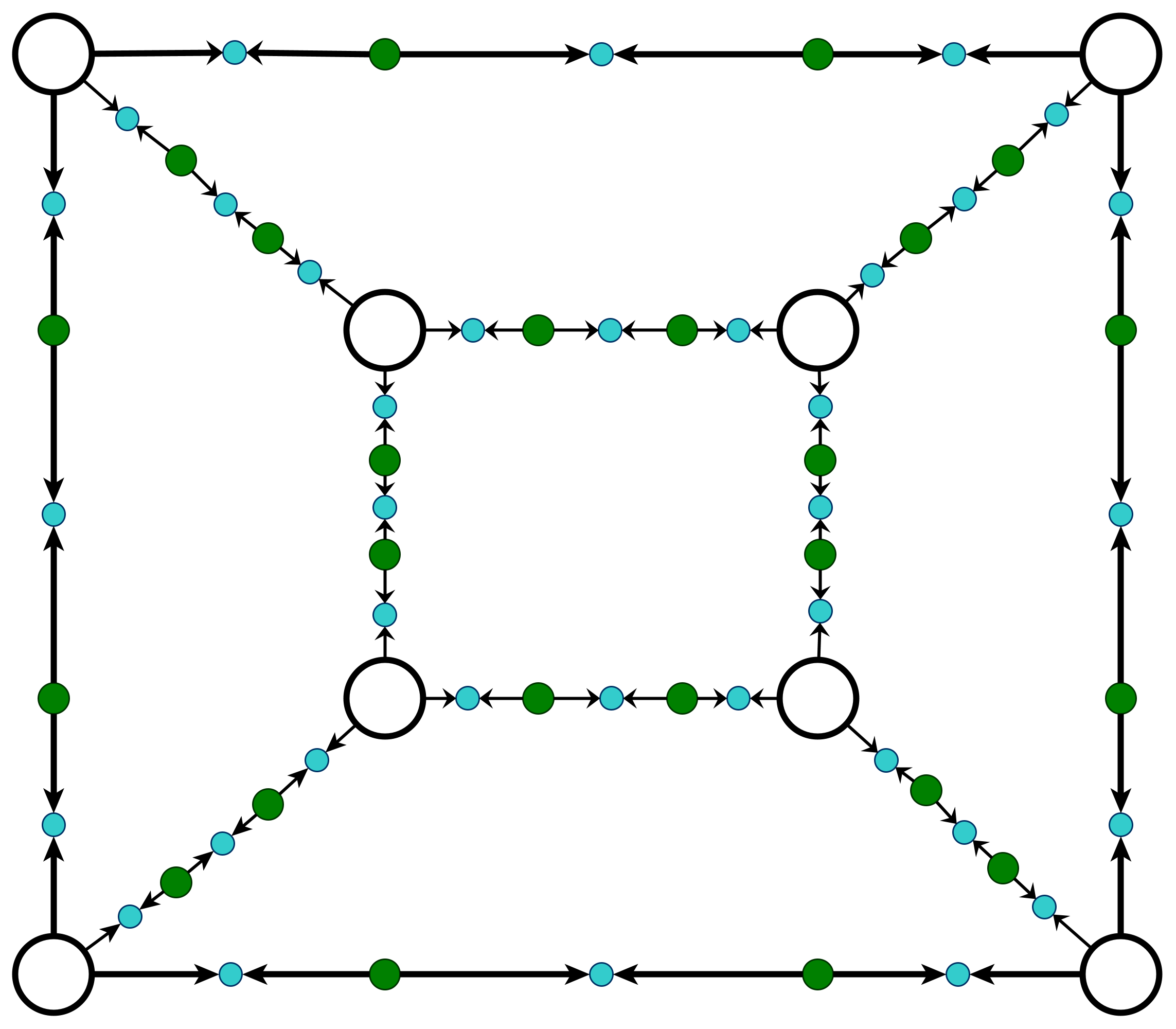} \caption{Blue vertices refer to $v_{e_i}^{cover}$ created in step $3$. For simplicity, vertices created in step 2 was omitted}\label{step3}\end{subfigure}
\hfill
\begin{subfigure}[b]{.48\linewidth}\centering \includegraphics[width=0.6\linewidth]{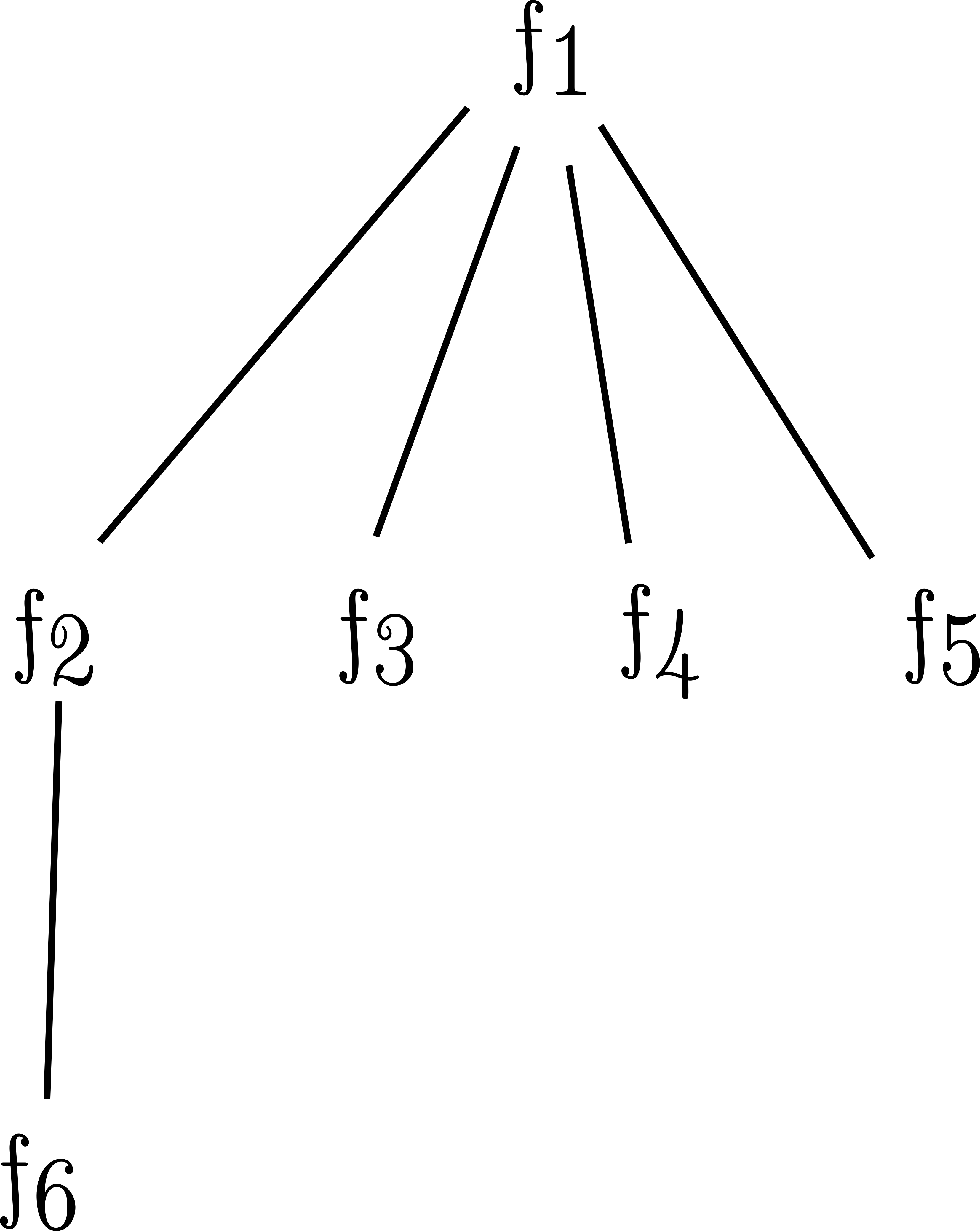} \caption{Spanning tree of the dual graph of $G$}
 \label{schemeT}\end{subfigure}
\caption{(a) Circuit after Step 3; (b) illustration of $\mathcal T_{D_G}$.}
\end{figure}

\end{enumerate}

Notice that $C$ is still planar. Now, preserving the planarity, we will ensure that every $v_{e_i}^{cover}$ outputs \texttt{true} for any assignment of $C$ as follows (see Fig.~\ref{fig:reduction}, Fig.~\ref{fig:reduction2} and Fig.~\ref{fig:reduction3}): 

\begin{enumerate}\setcounter{enumi}{3}
\item create an output vertex $v_{out}$ such that $f(v_{out}) = \texttt{AND}$;

\item for each vertex $v_{e_i}^{cover}$ which are in the external face of $G$, create one directed edge from $v_{e_i}^{cover}$ to $v_{out}$;
\end{enumerate}

Let $D_G$ be the dual graph of $G$, where $f_1$ represents the external face of $G$. Let $\mathcal T_{D_G}$ be the spanning tree of $D_G$ obtained from a breadth-first search of $D_G$ rooted at $f_1$. (see Fig.~\ref{schemeT})
In a top-down manner, according to a level-order traversal of $\mathcal T_{D_G}$, we visit each edge $e=(f_i,f_j)$ of $\mathcal T_{D_G}$ applying the following:

\begin{enumerate}\setcounter{enumi}{5}
\item let $f_j$ be a child of $f_i$ in $\mathcal T_{D_G}$; by construction of $G$, it follows that the boundary between $f_i$ and $f_j$ contains at least three edges, being at least one of which between vertices $a'$ and $b'$ that do not exist in $H$; thus, create a vertex $v_{f_j}$, add edges from $v_{f_j}$ to such $a'$ and $b'$; and for each $v_{e_\ell}^{cover}$ in the face $f_j$ that, yet, doesn't reach $v_ {out}$, add an edge from $v_{e_\ell}^{cover}$ to $v_{f_j}$; after that, if $v_{f_j}$ has in-degree greater than $0$ then set $f(v_{f_j}) = {\tt AND}$, otherwise create an input vertex $v_{f_j}^{in}$ add an edge from $v_{f_j}^{in}$ to $v_{f_j}$ and set $f(v_{f_j}) = {\tt AND}$;
\end{enumerate} 

\begin{figure}[!ht]
\centering
 \includegraphics[width=0.7\linewidth]{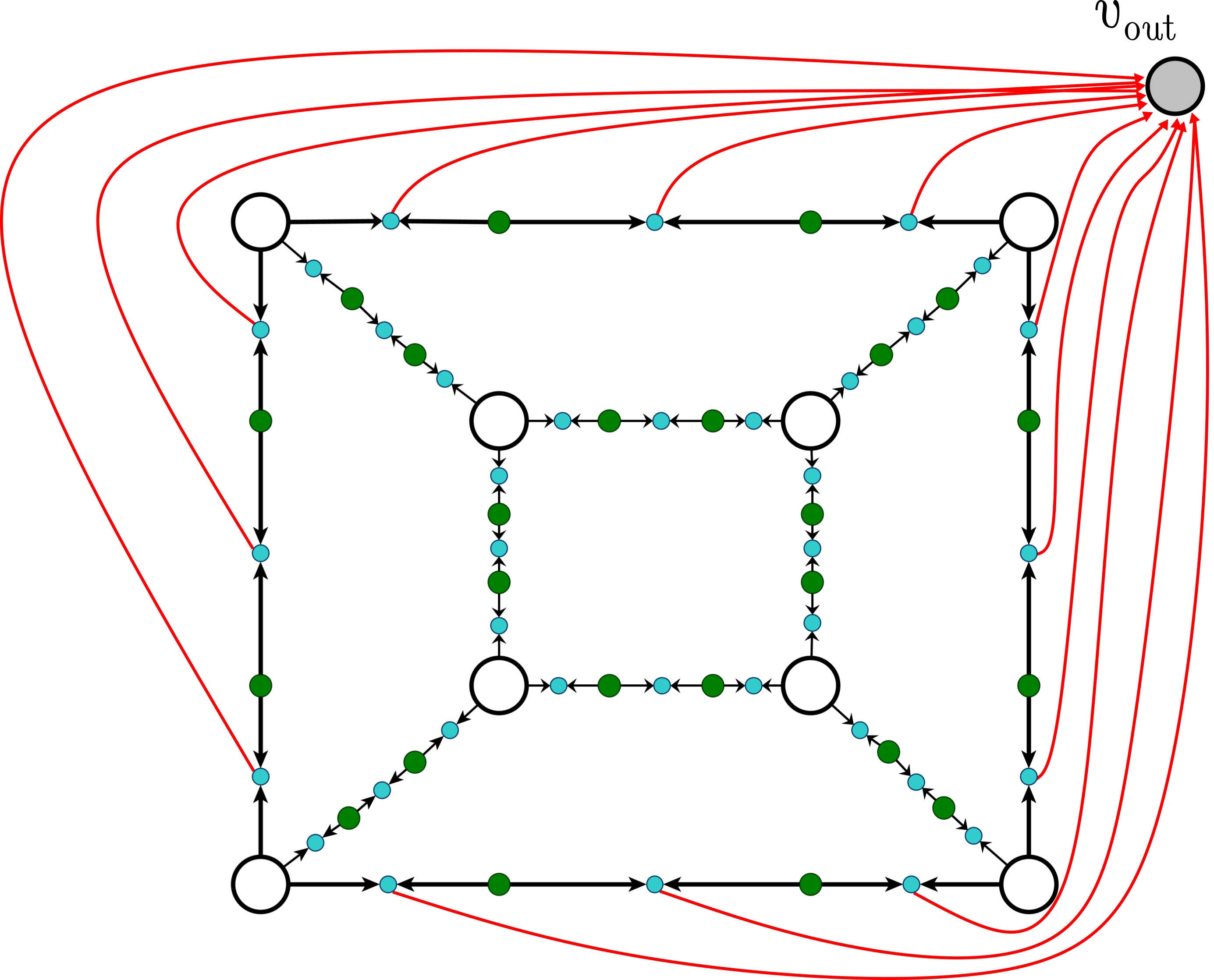}
\caption{Steps 4 and 5. Edges outgoing from $v_{e_i}^{cover}$ are highlighted in red. 
}\label{fig:reduction}
\end{figure}

\begin{figure}[!ht]
\centering
 \includegraphics[width=0.7\linewidth]{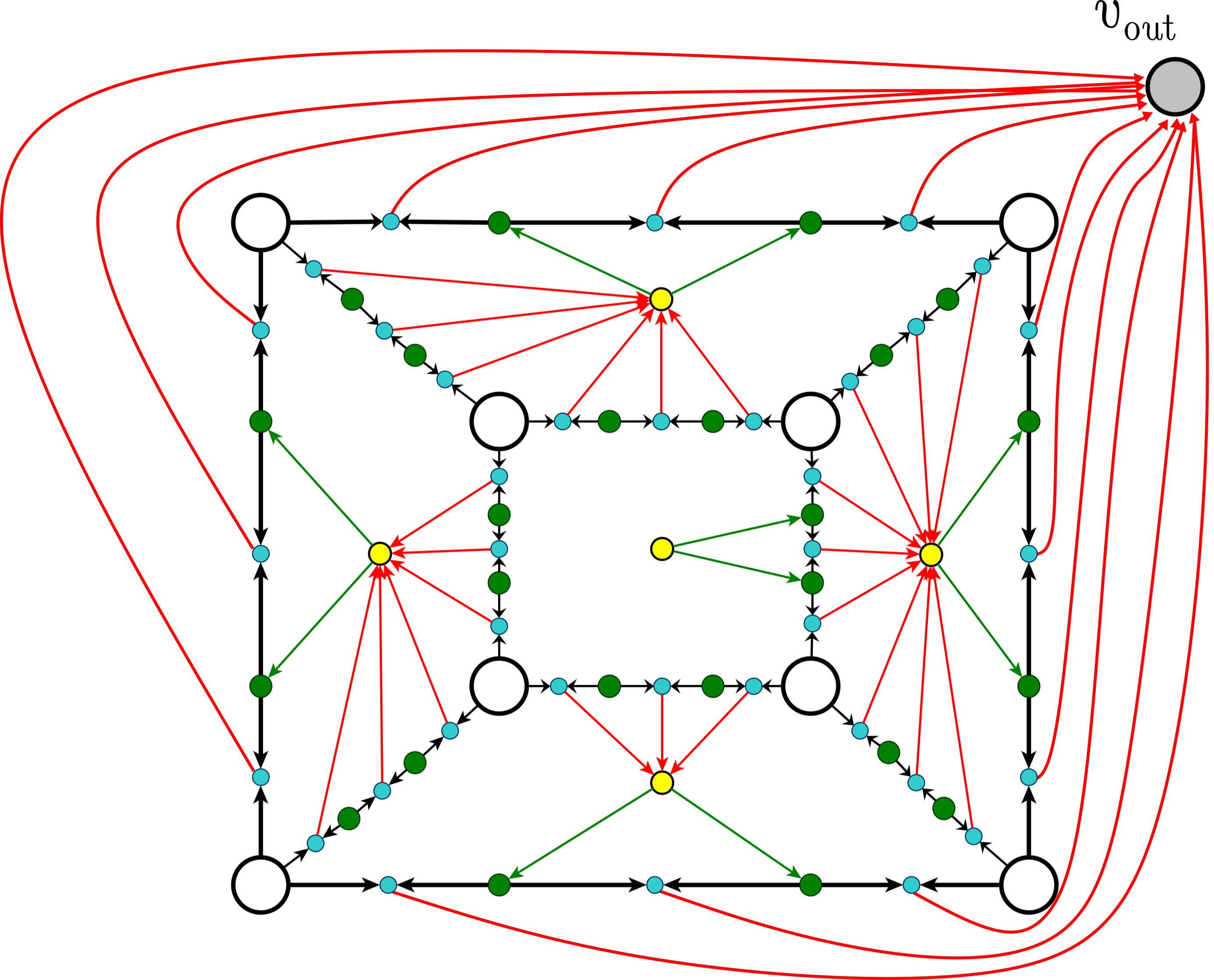}
\caption{Step 6. Each yellow vertex refers to $v_{f_j}$ inserted on step $6$; Edges outgoing from $v_{f_j}$ are highlighted in green (Inputs omitted).
}\label{fig:reduction2}
\end{figure}

From Step 6 holds that if $v_{out}$ outputs \texttt{true} then every vertex $v_{e_i}^{cover}$ and $v_{f_j}$ also output \texttt{true}. Besides, as each $v_{f_j}$ can be added in the planar embedding inside its respective face, then the resulting graph still planar, and by using $\mathcal T_{D_G}$ it holds that the added edges preserve the graph acyclic.

\begin{enumerate}\setcounter{enumi}{6}
\item finally, set $k = c + |E(G)| + |V(\mathcal {D_G})|$.
\end{enumerate} 

If $G$ has a vertex cover $S$ such that $|S| = c$ then we obtain a satisfying assignment $A$ to $C$ having energy complexity at most $k$ as follows: 
for each $v_i\in V(G)$ such that $v\notin S$ we set
\texttt{false} to its corresponding input $v_i^{in}$ (created in step 2 of the reduction); and we assign \texttt{true} to the other inputs.
Thus, in exactly $c$ edges flows \texttt{true} from vertices $v_ {i}^{in}$ to its out-neighbor $v_i$ (so, at most $c$ vertices $v_i$ will output \texttt{true});
and from each $v_i$ set as \texttt{true} flows positive values to each  $v_e^{cover}$ such that $e$ is an out-edge of $v_i$. Therefore, since $S$ is a vertex cover, in a bottom-up manner according to $\mathcal T_{D_G}$, we can observe that each face vertex $v_{f_j}$, each vertex $v_{e_i}^{cover}$, as well as $v_{out}$ will be set as \texttt{true}. Thus, $A$ is a satisfying assignment to $C$. At this point, it remains to analyze the energy of $A$. Note that $C$ has exactly $|E(G)|$ vertices $v_{e_i}^{cover}$, $|V({D_G})| - 1$ face vertices $v_{f_j}$, and all of them is set as \texttt{true} by $A$, which with the addition of $v_{out}$ implies energy consumption $|E(G)|+|V({D_G})|$, since at most $c$ vertices $v_i$ will output \texttt{true}, it follows that $A$ has energy $|E(G)|+|V({D_G})|+c$.

Conversely, let $C$ be a circuit with $k$ gates outputting \texttt{true}. By construction, if $C$ is satisfiable, then
all vertices $v_{e_i}^{cover}$ and $v_{f_j}$ also are satisfiable and, consequently, there is at most $c$ vertices $v_i$ outputting \texttt{true}, representing, in this way, a vertex cover with $c$ vertices in $G$.
\end{proof}


\begin{figure}[pht]
\begin{subfigure}[b]{.5\linewidth}\centering \includegraphics[width=0.8\linewidth]{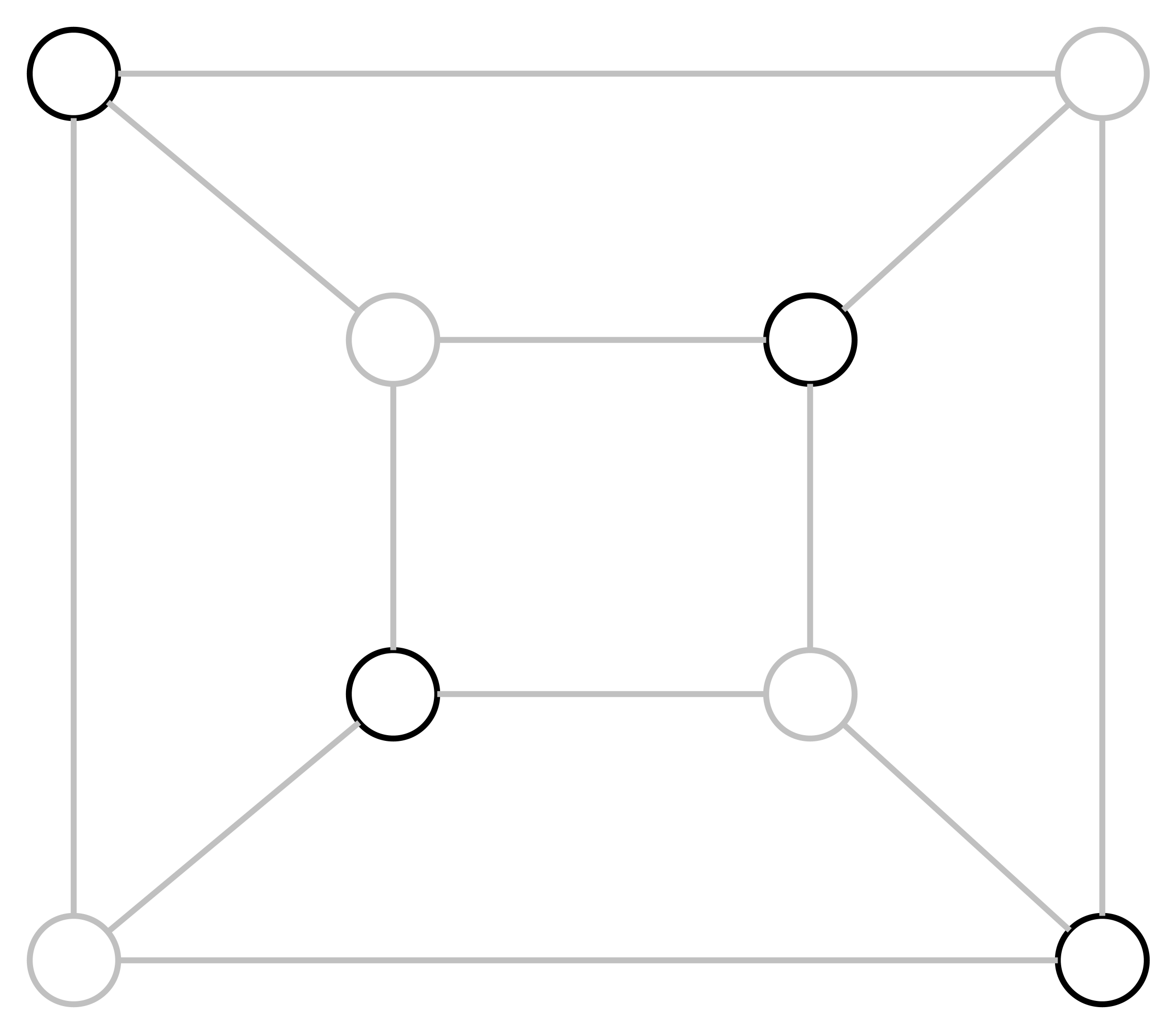} \caption{Vertex cover of $H$.}\label{fig:1vc}\end{subfigure}
\hfill
\begin{subfigure}[b]{.5\linewidth} \centering \includegraphics[width=0.8\linewidth]{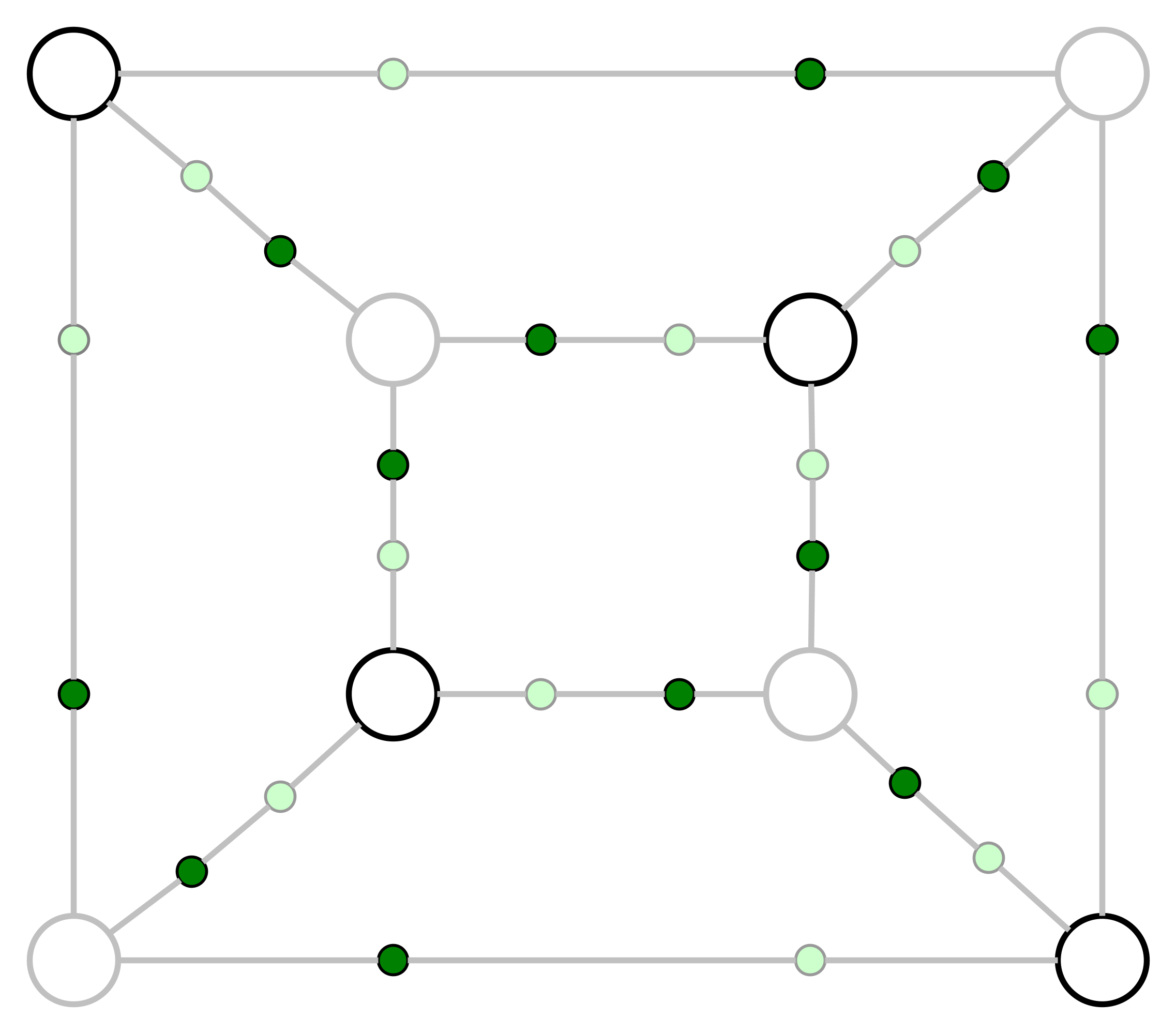} \caption{Corresponding vertex cover of $G$.}\label{fig:2vc}\end{subfigure}
\vfill
\begin{subfigure}[b]{\linewidth}\centering \includegraphics[width=0.65\linewidth]{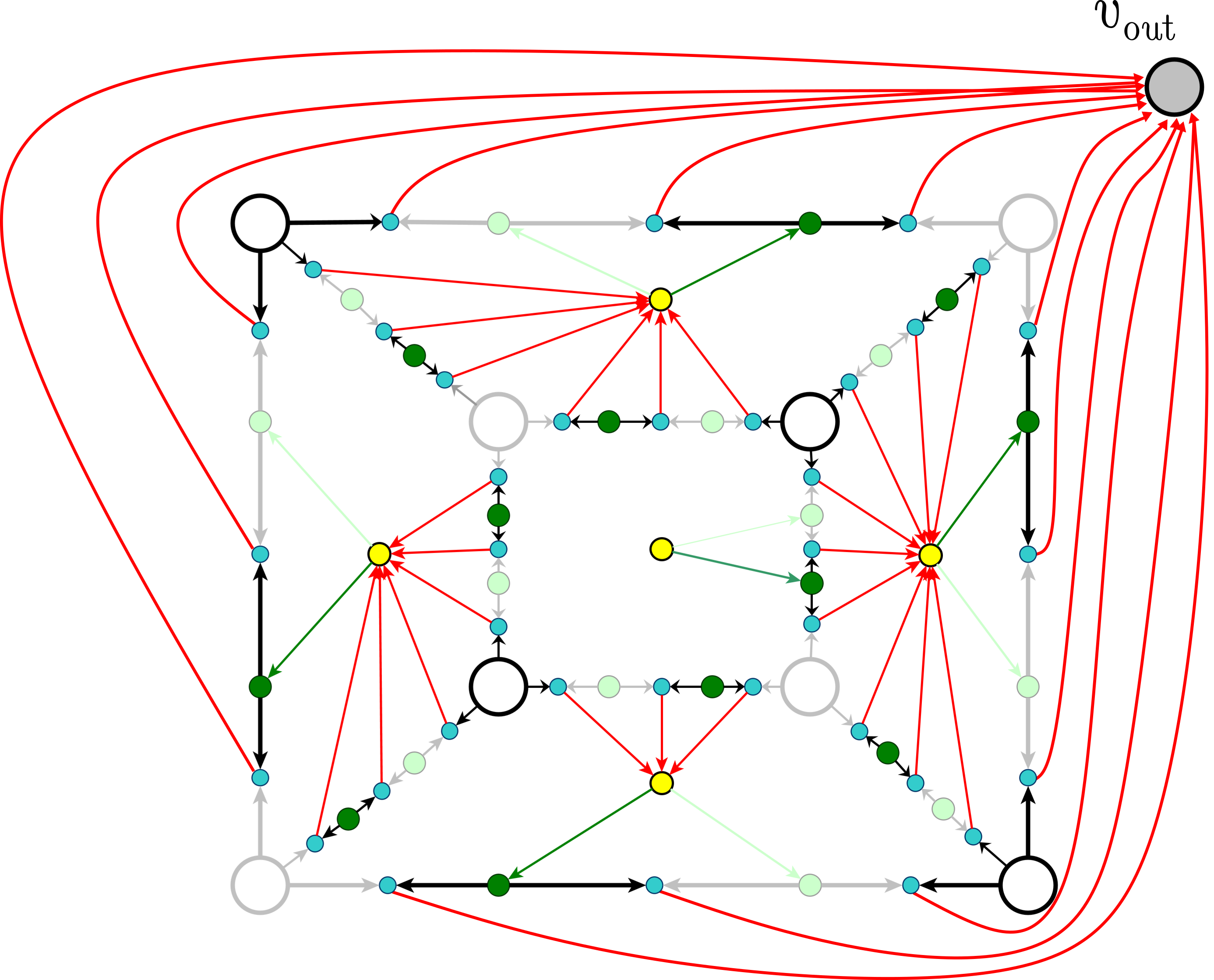} \caption{Accepting subtree of $G$ with only vertices enabled.}\label{fig:3vc}\end{subfigure}
\hfill
\begin{subfigure}[b]{\linewidth}\centering \includegraphics[width=0.9\linewidth]{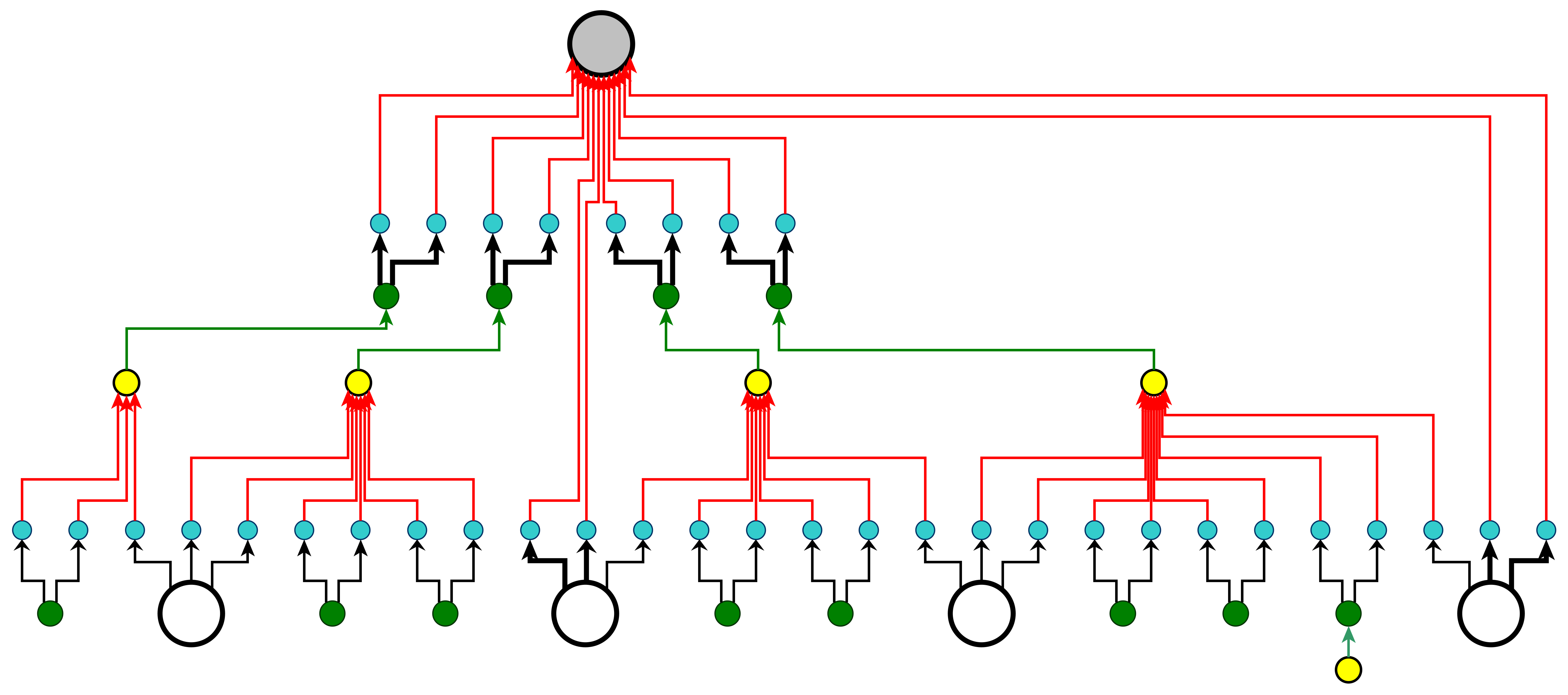} \caption{Bottom-up representation of accepting subtree.}\label{fig:4vc}\end{subfigure}
\vfill
\caption{Consequences of a minimum vertex cover of $H$ in the circuit $C$. Recall that each blue vertex is a {\tt OR}-vertex, $v_e^{cover}$, representing an edge $e$ of $G$, while the other vertices are {\tt AND}-vertices.}\label{fig:reduction3}
\end{figure}


\subsection{Parameterized Complexity}

Next, we investigate the parameterized complexity of \textsc{MinEC}$^+_M$.

\begin{theorem}\label{thm:XP}
$k$-\textsc{MinEC}$^+_M$ is in XP.
\end{theorem}

\begin{proof}

Let  $C = (V,E)$ be a circuit with $V=I \cup G\cup \{v_{out}\}$, where $I$ is the set of inputs of $C$, $G$ is the set of gates and $v_{out}$ is the output gate.
If $C$ has a satisfying assignment $X$ such that $EC(C,X)\leq k$ then we can find $X$ as follows:

\begin{enumerate}
    \item suppose that $X$ is the satisfying assignment with $EC(C,X)\leq k$ having minimum weight (i.e., minimum number of inputs assigned as {\tt true}); 

    \item first, we ``guess'' the set $T$ of gates that should be activated by $X$, that is, in $n^{O(k)}$ time, we enumerate each subset $T$ of gates such that $|T|\leq k$ and check each one in a new branch; 
    
    \item for each $T$ we can check in polynomial time whether it is consistent, that is: 
    \begin{itemize}
        \item $v_{out}\in T$;
        
        \item for each {\tt OR}-gate $v$ in $T$ either it has an in-neighbor in $T$ or it has an in-neighbor in $I$, and for each {\tt AND}-gate $v$ in $T$ its in-neighborhood is contained in $T\cup I$;
        
        \item conversely, each {\tt OR}-gate $w\notin T$ has no in-neighbor in $T$, and each {\tt AND}-gate $w\notin T$ has at least one in-neighbor that is not in $T$;
        
        \item also, no input is mutually in-neighbor of an {\tt AND}-gate in $T$ and an {\tt OR}-gate not in $T$;
        
        \item if $T$ is the set of gates activated by $X$ then it holds that: any input $i$ that is in-neighbor of an {\tt AND}-gate in $T$ should be set as {\tt true} in $X$; any input $i$ that is in-neighbor of an {\tt OR}-gate not in $T$ should be set as {\tt false} in $X$. Let $X'$ be such a partial assignment;
        
        \item at this point, for a consistent $T$, each {\tt OR}-gate $v$ in $T$ having no in-neighbor in $T$ has at least one in-neighbor in $I$ that is not set as {\tt false} by $X'$, and each {\tt AND}-gate $w\notin T$ having no in-neighbor in $G\setminus T$ has at least one in-neighbor in $I$ that is not set as {\tt true} in $X'$. 
    \end{itemize} 
    \item Since $X$ has minimum weight, from a given consistent set $T$, in order to extend $X'$ into a satisfying assignment $X$ with $EC(C,X)\leq k$ (if any), it is enough to ``guess'' the minimal set of inputs that should be set as {\tt true} to activate the {\tt OR}-gates in $T$ having no in-neighbor in $T$. As $|T|\leq k$, such subset of inputs is also bounded by $k$, thus, in $n^{O(k)}$ time, we can enumerate (if any) each assignment $X''$ extending $X'$ by setting at most $k$ additional inputs as {\tt true} in such a way that each {\tt OR}-gates in $T$ has at least one in-neighbor activated. At this point, from the guessed set $T$ we obtain the assignment $X$ if there is some $X''$ for which each {\tt AND}-gate $w\notin T$ having no in-neighbor in $G\setminus T$ has at least one in-neighbor in $I$ that is set as {\tt false}.   
\end{enumerate}

Note that for any satisfying assignment $X$ of $C$ the set of activated gates must satisfy the properties described in step 3. Since steps 2 and 4 check in $n^{O(k)}$ time all possibilities, it holds that \textsc{MinEC}$^+_M$ is XP-time solvable.
\end{proof}


Now, we show the W[1]-hardness of {\sc $k$-MinEC$^+_M$} using a reduction from \textsc{Multicolored Clique}.

\medskip
\noindent	\fbox{
		\parbox{0.96\textwidth}{
\noindent
{\sc \textsc{Multicolored Clique}}

\noindent
\textbf{Instance}: A graph $Q$ with a vertex-coloring $\ell : V(G) \rightarrow \{1,2, \dots , c \}$.

\noindent
\textbf{Parameter}: A positive integer $c$.

\noindent
\textbf{Question}: Does $Q$ have a $c$-clique containing all $c$ colors?     
}
}
\medskip


\begin{theorem}\label{w1hard}
{\sc $k$-MinEC$^+_M$} is W[1]-hard. 
\end{theorem}

\begin{proof}
Let $(Q,c)$ be  an instance of \textsc{Multicolored Clique} and let $V_1,V_2, \dots , V_c$ be the color classes of $Q$. Without loss of generality, we consider that each vertex in $V_i$ has at least one neighbor in $V_j (i \neq j)$. We construct an instance $(C,k)$ of {\sc MinEC$^+_M(k)$} as follows (see Fig.~\ref{fig:w1-Q} and Fig.~\ref{fig:w1-C}):

\begin{enumerate}
\item create an output gate $v_{out}$ in $C$ and set $f(v_{out}) = \texttt{AND}$;
\item for each color $c_i$ of $Q$, create a gate $w_i$ with $f(w_i) =$ \texttt{OR} and add an edge from $w_i$ to $v_{out}$;
\item for each color class $V_i$ of $Q$, create copies $V_{i}^{1}, V_{i}^{2},V_{i}^{3}$ and $V_{i}^{4}$ in C;
\item add edges from each vertex in $V_{i}^{4}$ to $w_i$;
\item let $v^1, v^2, v^3$ and $v^4$ be the copies of a vertex $v \in V(Q)$; add edges $(v^1, v^2), (v^2, v^3)$ and $(v^3, v^4)$ to $G$;
set $V_{i}^{1}$ as the input set; and assign $f(v^2) = f(v^3) = \texttt{OR}$ and $f(v^4) = \texttt{AND}$;
\item for each vertex $v^4 \in V_{i}^{4} (1 \leq i \leq c)$, create $c - 1$ new \texttt{OR}-in-neighbors $a_{v^4}^{j} (1 \leq j \leq c$ and $i \neq j)$, and for each $u^2 \in V_{j}^{2}$ such that $vu \in E(Q)$ create an {\tt AND}-vertex $b_{vu}^j$ and the following edges: ($b_{vu}^j,a_{v^4}^{j}$), ($u^2,b_{vu}^j$) and ($v^2,b_{vu}^j$);


\item finally, set $k = 2c^2 + 2c + 1$.
\end{enumerate}


\begin{figure}[ht]
    \centering
    \includegraphics[width = 4cm]{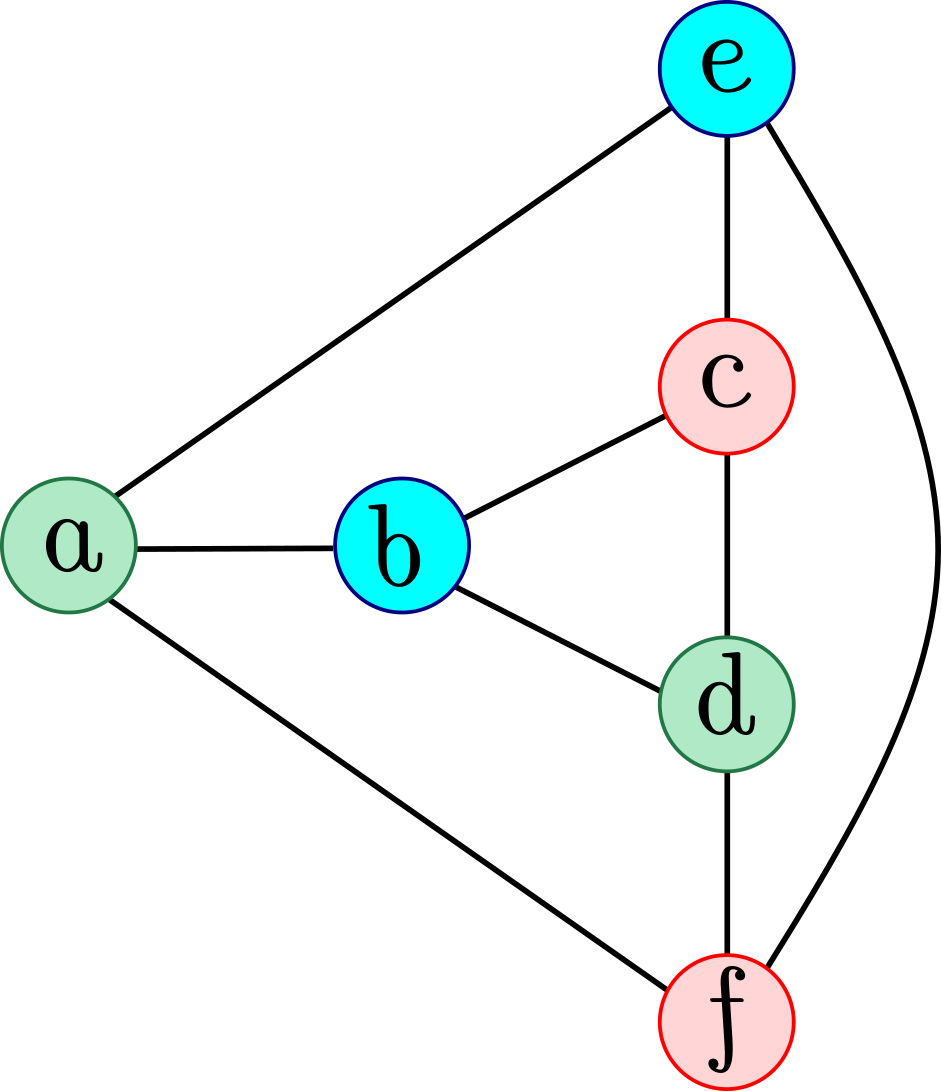}
    \caption{An instance $Q$ for  {\sc Multicolored Clique}}
    \label{fig:w1-Q}
\end{figure}


\begin{figure}[ht]
    \centering
    \includegraphics[width=\textwidth]{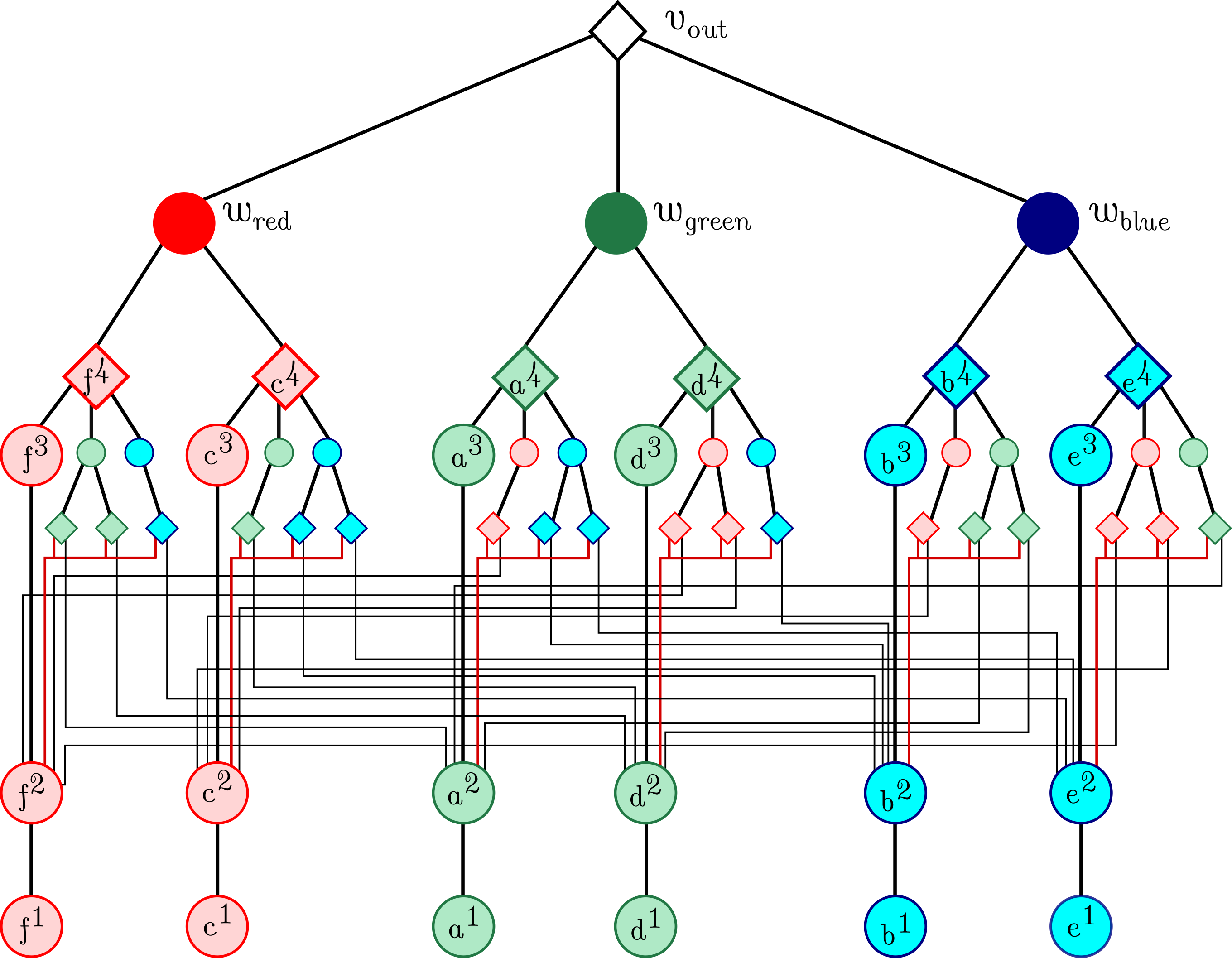}
    \caption{Circuit $C$ obtained from $Q$ (Fig. \ref{fig:w1-Q}) after reduction. The vertices represented as rhombuses are {\tt AND}-gates; the other vertices are {\tt OR}-gates except those with in-degree $0$ (inputs).}
    \label{fig:w1-C}
\end{figure}



If $Q$ contains a multicolored clique $K$ such that $|K| = c$, then it is possible to find a satisfying assignment of $C$ that consumes $k$ energy by mapping the set $S$ of gates/vertices that must be activated (outputs \texttt{true}) as follows: 
(a) $v_{out}$ and all of its in-neighbors must belong to $S$; 
(b) for each \texttt{OR}-gate $w_i \in S$, we want include in $S$ exactly the in-neighbor $v^4\in V^4_i$ such that $v \in K$, therefore, we set $f(v^1)$=\texttt{true} if and only if $v \in K$ (At this point, by construction, for each $v \notin K$ holds that every vertex between $v^1$ and $v^4$ will be inactivated);
(c) for each $v^4 \in S$, all of its in-neighbors must be in $S$, and for each $a^j_{v^4}$ in $S$, its unique in-neighbor in $S$ must be the \texttt{AND}-gate $b_{vu}^j$ such that $f(v^1)=f(u^1)$=\texttt{true} (recall that $K$ has exactly one vertex per color).
(d) finally, a vertex $v^2\in V_2$ belongs to $S$ if and only if its in-neighbor $v^1$ outputs true. 
Through a simple count one can conclude that $|S|=2c^2 + 2c + 1$. Thus, the defined assignment satisfies $C$ by consuming $k$ energy as required.

Conversely, if $C$ has a satisfying assignment $X$ with energy complexity at most $k=2c^2 + 2c + 1$ then it is possible to obtain a multicolored clique $K$ of $Q$ as follows: a vertex $v$ of $Q$ belongs to $K$ if and only if $v^2$ outputs \texttt{true}. 
Since, by construction, any satisfying assignment of $C$ activates at least $2c^2 + 2c + 1$ gates in $V(C)\setminus(V_2\cup V_1)$, the assignment $X$ activates at most $c$ gates in $V_2$. Besides, the construction also implies that at least one input per color must activated in order to satisfy $C$. So, $X$ activates exactly $c$ gates in $V_2$ (one per color). Therefore, $K$ has exactly one vertex per color. Now, to show that $K$ induces a clique is enough to observe the if $v_2$ and $u_2$ are activated in $X$ into $C$ and $vu\notin E(G)$, then for the color $j$ of $u$ holds that $b^j_{vr}$ is inactivated by $X$ into $C$ for any neighbor $r$ of $v$ with color $j$. Thus, $v_4$ and $w_i$ are also inactivated, where $i$ is the color of $v$, which contradicts the fact that $X$ satisfies $C$. Therefore, $K$ induces a clique.
\end{proof}

\subsection{On monotone circuits with bounded genus}\label{sec:genus}



A graph $G$ has \emph{genus} at most $g$ if it can be drawn on a surface of genus $g$ (a sphere with $g$ handles) without edge intersections (see Fig.~\ref{fig:genus}). We consider the genus of a circuit as the genus of its underlying undirected graph. Additionally, We refer the reader to~\cite{gross2001topological} for more information on the genus of a graph.

\begin{figure}[!ht]
\begin{subfigure}[b]{.32\linewidth}\centering \includegraphics[width=\linewidth]{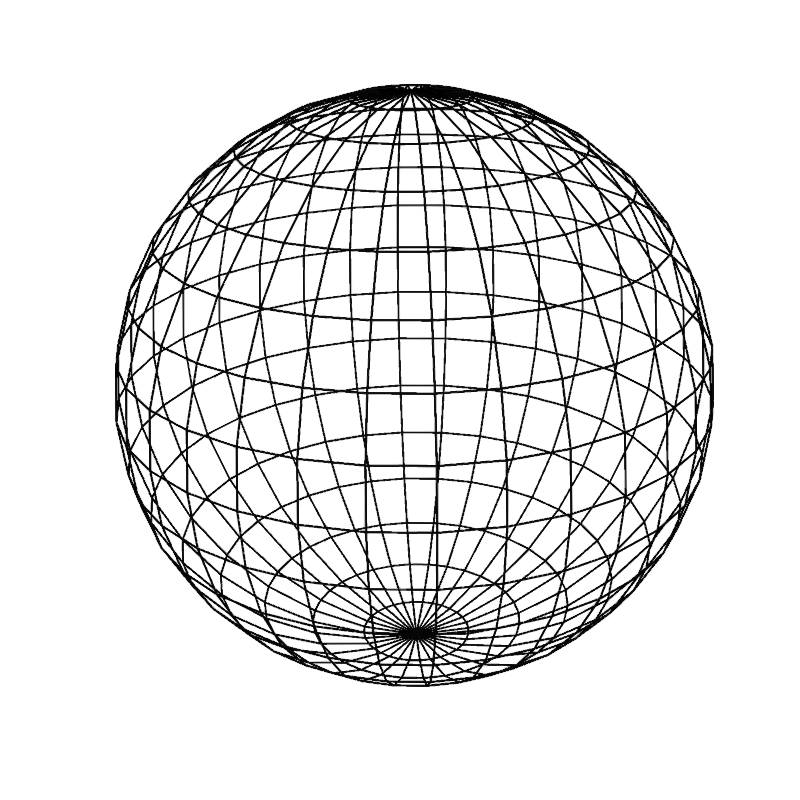} \caption{Genus 0.}\label{fig:genus0}
\end{subfigure}
\hfill
\begin{subfigure}[b]{.32\linewidth}\centering \includegraphics[width=\linewidth]{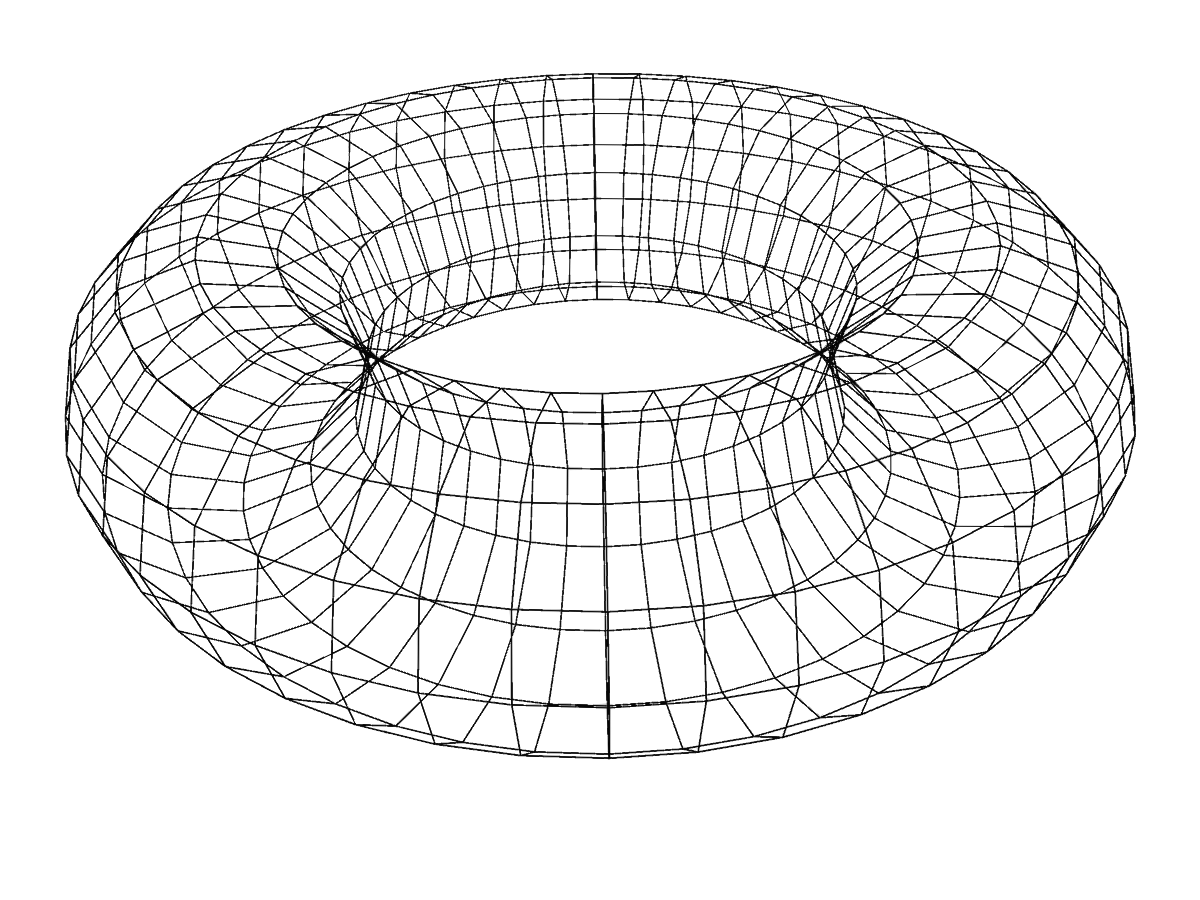} \caption{Genus 1.}\label{fig:genus1}
\end{subfigure}
\hfill
\begin{subfigure}[b]{.32\linewidth}\centering \includegraphics[width=\linewidth]{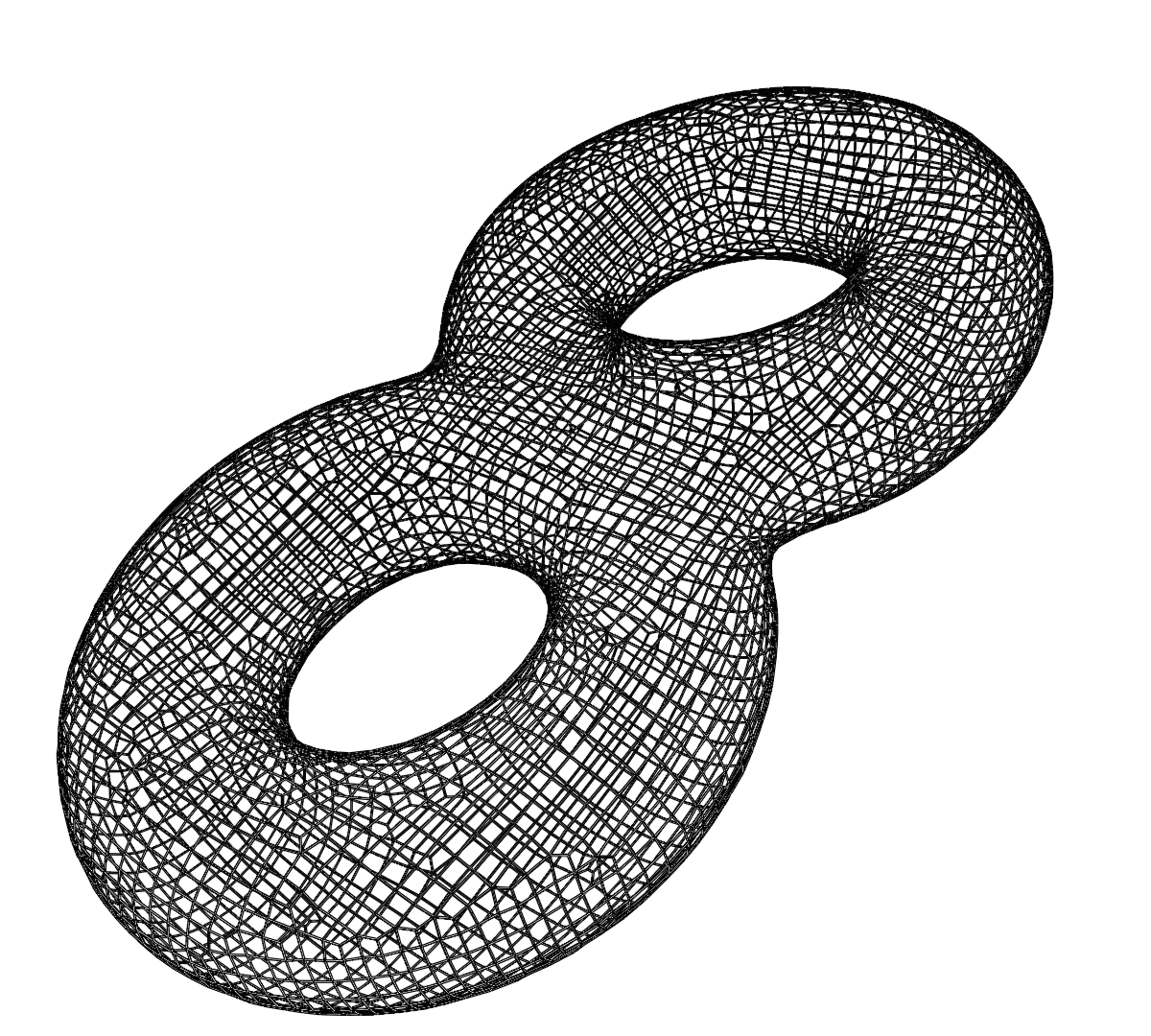} \caption{Genus 2}\label{fig:genus2}
\end{subfigure}
\caption{Relation between genus and geometrical shapes.}\label{fig:genus}
\end{figure}

In this section, we show that {\sc $k$-MinEC$^+_M$} on bounded genus circuits can be reduced to {\sc $k$-MinEC$^+_M$} on bounded treewidth circuits.

\begin{definition}
A \emph{tree decomposition} of a undirected graph $G$ is a pair $\mathcal{T} = (T,\{X_t\}_{t \in V(T)})$, such that $T$ is a tree where each node $t$ is assigned to a set of vertices $X_t \subseteq V(G)$, called \emph{bags}, according to the following conditions:
\begin{itemize}
\item $\bigcup_{t \in V(T)} X_t = V(G)$;
\item For each $uv \in E(G)$ there is a node $t$ such that $\{u,v\} \subseteq X_t$;
\item For each $v \in V(G)$, the set $T_v = \{t \in V(T) : v \in X_t\}$ spans a subtree of $T$.
\end{itemize}
\end{definition}

The \emph{width} of a tree decomposition $\mathcal{T}$ is the size of its largest bag minus one. The treewidth of $G$ is the minimum width among all tree decompositions of $G$.

\begin{definition}
A  graph $H$ is a \emph{minor} of a graph $G$ if $H$ can be constructed from $G$ by deleting vertices or edges, and contracting edges.
\end{definition}


\begin{theorem}[Excluded Grid Theorem~\cite{robertson1994quickly}]\label{theo:grid}
Let $t$ be a non-negative integer. Then every planar graph $G$ of treewidth at least $9t/2$ contains a grid $t \times t$ as a minor.
\end{theorem}

From the Excluded Grid Theorem, it is easy to see that there is a connection between the diameter of a planar graph and its treewidth. In~\cite{robertson1984graph}, Robertson and Seymour presented a bound for the treewidth of a planar graph with respect to its radius, which also implies a bound regarding the diameter.

\begin{definition}
For every face $F$ of a planar embedding $M$, we define $d(F)$ to be the minimum value of $r$ such that there is a sequence $F_0,F_1,\ldots,F_r$ of faces of $M$, where $F_0$ is the external face, $F=F_r$, and for $1\leq j\leq r$ there is a vertex $v$ incident with both $F_{j-1}$ and $F_{j}$. The radius $\rho(M)$ of $M$ is the minimum value $r$ such that $d(F)\leq r$ for all faces $F$ of $M$. The radius of a planar graph is the minimum of the radius of its planar embeddings.
\end{definition}

\begin{theorem}[Radius Theorem~\cite{robertson1984graph}]\label{theo:radius}
If $G$ is planar and has radius at most $r$ then its treewidth is at most $3r+ 1$.
\end{theorem}

Using Theorem~\ref{theo:radius} we are able to either solve {\sc MinEC$^+_M$} on planar circuits or outputs an equivalent instance $C'$ with treewidth bounded by a function of $k$. 



\begin{lemma}\label{theo:twbounded}
Let $(C,k)$ be an instance of {\sc MinEC$^+_M$}.
There is an algorithm that in polynomial time either solves $(C,k)$ or outputs an equivalent instance $(C',k)$ of {\sc MinEC$^+_M$} where the distance from $v_{out}$ to each vertex in the underlying graph of $C'$ is at most $2k + 1$.
\end{lemma}


\begin{proof}
From an instance $(C,k)$ of {\sc MinEC$^+_M$}, we apply the following reduction rules to obtain $C'$:
\begin{enumerate}
\item Delete every input vertex which is at a distance greater than $k$ to $v_{out}$;
\item Delete every vertex which is at a distance greater than $k + 1$ from its nearest input vertex;
\item Delete any {\tt AND}-vertex which lost one of its in-neighbors;
\item Delete any {\tt OR}-vertex in which its in-degree became equal to $0$;
\item Repeat steps 1 to 4 as long as possible;
\item If $C'=\emptyset$ then we conclude that $(C,k)$ is a {\em no}-instance of {\sc MinEC$^+_M$}.
\end{enumerate}

We now discuss the safety of the previous reduction rules: if an input vertex $v$ is at a distance greater than $k$ from $v_{out}$, since $C$ is monotone, then $v$ is not useful to satisfy $v_{out}$ in any assignment $X$ with $EC(C,X)\leq k$, thus we can assume that $v$ outputs \texttt{false} and given the monotonicity of $C$ we can safely remove $v$ (Rule 1). Similarly, gates that are at a distance greater than $k$ from its nearest input vertex must output \texttt{false} in an assignment $X$; otherwise, $X$ consumes energy greater than $k$. Note that vertices at a distance exactly $k+1$ from its nearest input vertex can be useful to show that a given assignment consumes energy greater than $k$. However, gates at a distance of at least $k+2$ from its nearest input vertex can be removed once its neighbors are sufficient to certify the negative answer (Rule 2). Besides, if for any assignment $X$ with $EC(C,X)\leq k$ holds that some (resp. every) in-neighbor of an {\tt AND}(~resp. {\tt OR})-vertex $v$ must output {\tt false}, then $v$ must output {\tt false} as well. Thus, Rule 3 and Rule 4 are safe. From the safety of rules 1-4, it follows that Rule 5 and Rule 6 are safe. 
Finally, if $C'\neq \emptyset$ then $C'$ has only vertices at a distance at most $2k + 1$ from $v_{out}$ in the underlying undirected graph of $C'$.
\end{proof}


Note that the underlying undirected graph of the circuits obtained from Lemma~\ref{theo:twbounded} have diameter bounded by $4k + 2$. Therefore, contrasting with the W[1]-hardness for the general case, Corollary~\ref{cor:indegree} holds.

\begin{corollary}\label{cor:indegree}
{\sc $k$-MinEC$^+_M$} is fixed-parameter tractable when restricted to monotone circuits having bounded maximum in-degree. 
\end{corollary}
\begin{proof}
First, apply the preprocessing algorithm presented in Lemma~\ref{theo:twbounded}.  After that, each input vertex is at a distance of at most $k$ from $v_{out}$ in the resulting circuit $C'$. Since each input vertex of $C'$ must reach $v_{out}$, $C'$ is acyclic, and it has bounded maximum in-degree, then from a BFS algorithm having $v_{out}$ as root (considering the reverse order of edges) we can visit all input vertices using at most $k$ levels, implying that the number of input vertices in $C'$ is bounded by $c^k$, where $c$ is its maximum in-degree. Also, each gate of $C'$ is at distance at most $k+1$ from its nearest input vertex, and since the number of input vertices in $C'$ is bounded, by a similar argument, the number of gates of $C'$ is also bounded by a function of $k$. 
Thus, when the input circuit $C$ has bouded maximum in-degree, after the preprocessing of Lemma~\ref{theo:twbounded}, the resulting circuit $C'$ is a kernel for {\sc $k$-MinEC$^+_M$}, implying its fixed-parameter tractability.   
\end{proof}

Notice that a gate with large in-degree can always be replaced by a binary tree using only binary gates, but for or-gates it makes a relevant difference in the energy complexity. Therefore, replacing large in-degree gates is not a useful strategy for dealing with {\sc $k$-MinEC$^+_M$}.
On the other hand, Lemma~\ref{theo:twbounded} also implies that if $C'$ is planar then it also has bounded radius
, thus, by Theorem~\ref{theo:radius}, it follows that the underlying undirected graph of $C'$ has treewidth bounded by a function of $k$.
We extend the previous reasoning for bounded genus circuits.

Given a vertex-set $S\subseteq V(G)$ of a simple graph $G$ such that the subgraph of $G$ induced by $S$, denoted $G[S]$, is connected, contracting $S$ means contracting the edges between the vertices in $S$ to obtain a single vertex at the end. We say that a graph $H$ is an {\em s-contraction} of a graph $G$ if $H$ can be obtained after applying to $G$ a (possibly empty) sequence of edge contractions.

\begin{figure}[!ht]
    \centering
    \includegraphics[width=0.7\textwidth]{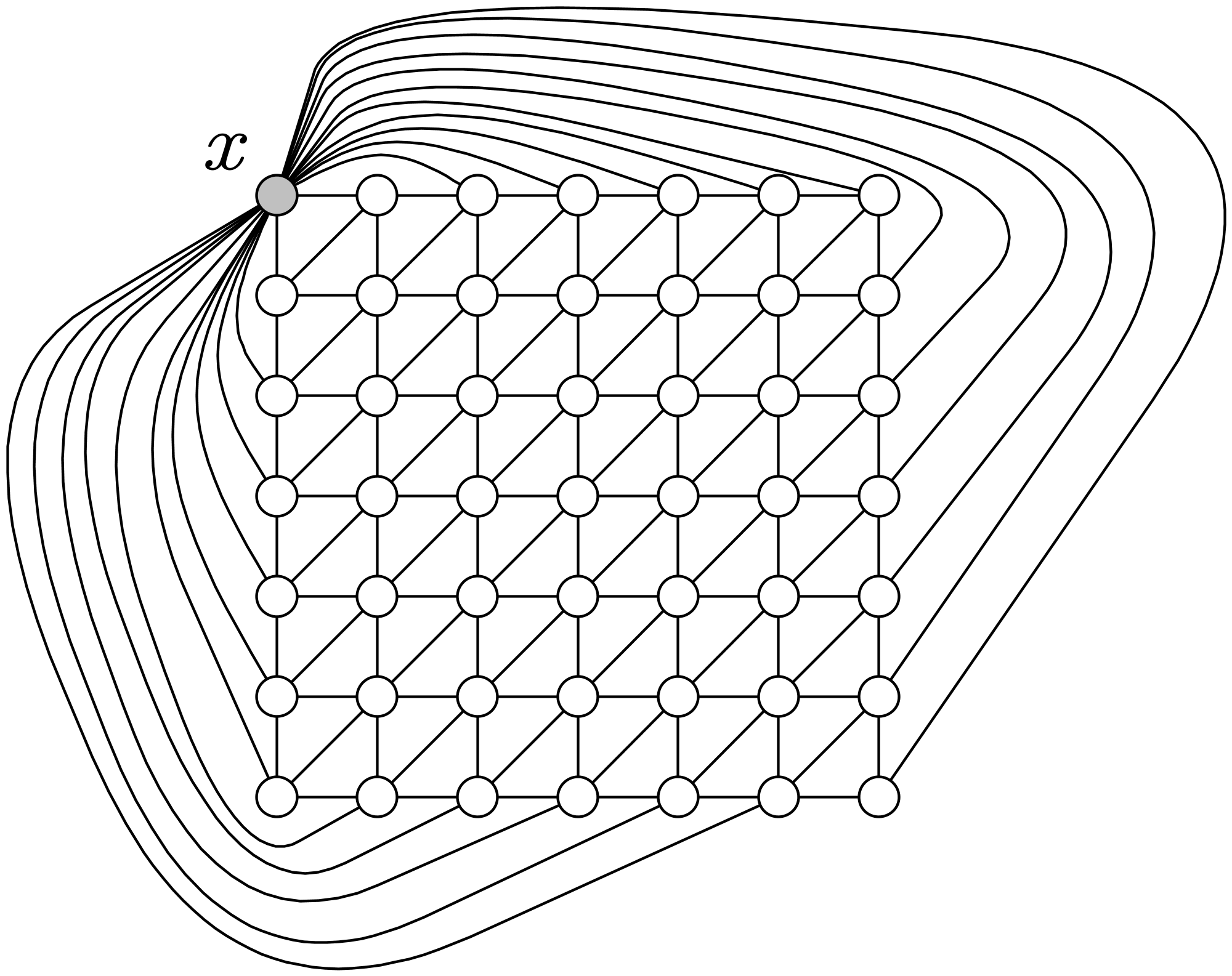}
    \caption{$\Gamma_7$ representation.}
    \label{fig:Gamma}
\end{figure}

The following is a construction presented in~\cite{fomin2011contraction} and \cite{kanj2017parameterized}. Consider an $(r\times r)$-grid. A corner vertex of the grid is a vertex of the grid of degree 2. By $\Gamma_r$ we denote the graph obtained from the $(r\times r)$-grid as follows (see Fig.\ref{fig:Gamma})
: construct first the $\Gamma'_r$ by triangulating all internal faces of the $(r\times r)$-grid such that all internal vertices of the grid are of degree 6, and all non-corner external vertices of the grid are of degree 4 ($\Gamma'_r$ is unique up to isomorphism). Two of the corners of the initial grid have degree 2 in $\Gamma'_r$; let $x$ be one of them. $\Gamma_r$ obtained from $\Gamma'_r$ by adding all the edges having $x$ as an endpoint and a vertex of the external face of the grid that is not already a neighbor of $x$ as the other endpoint. Observe again that $\Gamma'_r$ is unique up to isomorphism. The following is a lemma from~\cite{kanj2017parameterized} implied from Lemma 6 in~\cite{fomin2011contraction}.

\begin{lemma}[Lemma 4.5 in~\cite{kanj2017parameterized}]\label{genus}
Let $G$ be a graph of genus $g$, and let $r$ be any positive integer. If $G$ excludes $\Gamma_r$ as an s-contraction, then the treewidth of $G$ is at most $(2r+4)\cdot(g+1)^{3/2}.$
\end{lemma}

\begin{lemma}\label{treewidthGenus}
Let $C'$ be the circuit obtained from Lemma~\ref{theo:twbounded}. It holds that $C'$ has treewidth at most $(8k + 14)\cdot(g+1)^{3/2}$, where $g$ is the genus of $C'$.
\end{lemma}
\begin{proof}
First, notice that for each vertex $u$ of a $\Gamma_{4k+5}$ there is another vertex $v$ such that the distance between $u$ and $v$ is at least $2k+2$.
Now, suppose that $C'$ has $\Gamma_{4k+5}$ as an s-contraction, and let $u$ be a vertex of a $\Gamma_{4k+5}$ such that $u$ is either $v_{out}$ or a vertex obtained by contracting $S$ containing $v_ {out}$. Since there is a vertex $v$ such that the distance between $u$ and $v$ is at least $2k+2$, it holds that $C'$ has a vertex at distance greater than $2k+1$ from $v_{out}$, which is a contradiction (see Lemma~\ref{theo:twbounded}). Thus, by Lemma~\ref{genus} we have that the treewidth of $C'$ is at most $(8k + 14)\cdot(g+1)^{3/2}$.
\end{proof}

\section{Dynamic programming on bounded treewidth circuits}

From Lemma~\ref{treewidthGenus}, in order to solve {\sc $k$-MinEC$^+_M$} in FPT-time on bounded genus ins\-tances, it is enough to present an FPT algorithm parameterized by the treewidth of the input. To design a dynamic programming on tree decompositions, without loss of generality, we may consider that
we are given a tree decomposition that is a rooted \emph{extended nice tree decomposition} (see~\cite{cygan2015parameterized} for details), which is defined as follow.

\begin{definition}
A tree decomposition $\mathcal{T}$ is an \emph{extended nice tree decomposition} if the following conditions are satisfied:
\begin{itemize}
 \item The root bag $X_r$ and the leaf bags are empty;
 
 \item Every non-leaf node of $\mathcal{T}$ is one of the types described below:
 \begin{itemize}
 \item Introduce vertex node -- a node $t$ with exactly one child $t'$ such that $X_t = X_{t'} \cup \{v\}$ for some $v \notin X_{t'}$, we say that $v$ is introduced in $t$;
 
 \item Introduce edge node -- a node $t$, labeled with an edge $uv \in E(G)$ such that $\{u,v\} \in X_t$, and with exactly one child $t'$ such that $X_t$ = $X_{t'}$, we say that edge $uv$ is introduced at $t$;
 
 \item Forget node -- a node $t$ with exactly one child $t'$ such that $X_t = X_{t'} \setminus \{v\}$ for some $v \in X_{t'}$, we say that $v$is forgotten at $t$;
 
 \item Join node -- a node $t$ with two children $t'$ and $t''$ such that $X_t = X_{t'} = X_{t''}$.
 \end{itemize}
\item every edge of $E(G)$ is introduced exactly once in the whole decomposition.
\end{itemize}
\end{definition}

Based on the following results, we can assume that we are given a nice tree decomposition of $G$ without loss of generality.

\begin{theorem}~\cite{bodlaender2016c}
There exists an algorithm that, given an $n$-vertex graph $G$ and an
integer $k$, runs in time $2^{O(k)}\cdot n$ and either outputs that
the treewidth of $G$ is larger than $k$ or constructs a tree
decomposition of $G$ of width at most $5k + 4$.
\end{theorem}

\begin{lemma}~\cite{cygan2015parameterized}
Given a tree decomposition $(T, \{X_t\}_{t\in V(T)})$ of $G$ of width
at most $k$, one can in time $O(k^2\cdot \max(|V(T)|,|V(G)|))$ compute
a nice tree decomposition of $G$ with at most
$O(k\cdot|V(G)|)$ nodes and width at most $k$.
\end{lemma}

Now, we are ready to use a nice tree decomposition to obtain
an FPT-time algorithm for {\sc MinEC$^+_M$} parameterized by $tw(G)$, with single exponential dependency on $tw(G)$ and linear with respect to $n$.

\begin{theorem}\label{theo:dp}
{\sc MinEC$^+_M$} can be solved in time $2^{O(tw)} \cdot n$, where $tw$ is the treewidth of the underlying undirected graph of the input.
\end{theorem}

\begin{proof}
Let $C = (I,G,v_{out})$ be a monotone circuit where $I$ is the set of inputs of $C$, $G$ is the set of gates with out-degree greater than $0$ and $v_{out}$ is a single output vertex. Let $\mathcal{T} = (T,\{X_t\}_{t \in V(T)})$ be a rooted extended nice tree decomposition of $C$. Consider also $T_{t}$ as the subtree of $T$ rooted by node $t$ (bag $X_t$) and $C_{t}$ be the graph/circuit having $\mathcal{T}_{t} = (T_{t},\{X_i\}_{i \in V(T_{t})})$ as tree decomposition.
For convenience, we add the vertex $v_{out}$ to every bag of $T$; thus, the width of $\mathcal{T}$ is increased by at most one. 

Now, note that an assignment $X$ satisfies a monotone circuit $C$ if and only if it induces an activation set $\mathcal{S}_X$ such that:
\begin{enumerate}
    \item $v_{out}\in \mathcal{S}_X$;
    \item for each $v\in \mathcal{S}_X$ holds that:
    \begin{itemize}
        \item if $f(v)=${\tt AND} then every in-neighbor of $v$ is in $\mathcal{S}_X$;
        \item if $f(v)=${\tt OR} then at least one among the in-neighbors of $v$ is in $\mathcal{S}_X$;
    \end{itemize}
    \item for each $v\notin \mathcal{S}_X$ holds that:
    \begin{itemize}    
        \item if $f(v)=${\tt AND} then at least one among the in-neighbor of $v$ is not in $\mathcal{S}_X$;
        \item if $f(v)=${\tt OR} then every in-neighbor of $v$ is not in $\mathcal{S}_X$;
    \end{itemize}
\end{enumerate}

Properties 1 and 2 describe the necessary and sufficient conditions for a set $\mathcal{S}_X$ of activated gates to certify a satisfying assignment. Property 3 ensures that $\mathcal{S}_X$ is maximal regarding the property of having been activated by $X$.

Therefore, the problem of finding a satisfying assignment $X$ which minimizes $EC(C,X)$ can be seen as the problem of finding a satisfying assignment $X$ which minimizes $|\mathcal{S}_X \setminus I|$.
Thus, we define $c[t,S,\mathcal{B}^{\tt OR},\mathcal{B}^{\tt AND}]$ as 
the cardinality of a minimum set of gates $\mathcal{S}_t$ (if any) 
of $C_{t}$ such that:
\begin{itemize}
    \item $v_{out}\in S$ and $S= X_t\cap \mathcal{S}_t$;
    (we say that $X_t \setminus S = \overline{S}$)
    
    \item for each $v\in V(C_t)\setminus X_t$ properties 2 and 3 holds with respect to $\mathcal{S}_t$;
    
    \item for each $v\in S$ such that $f(v)=${\tt AND}, all in-neighbors of $v$ in $C_t$ are in $\mathcal{S}_t$;
    
    \item The set $\mathcal{B}^{\tt OR}$ is the subset of {\tt OR}-gates in $S$ already having in-neighbors in $\mathcal{S}_t$;
    
    \item for each $v\in \overline{S}$ such that $f(v)=${\tt OR}, all in-neighbors of $v$ in $C_t$ are not in $\mathcal{S}_t$;
    
    \item The set $\mathcal{B}^{\tt AND}$ is the subset of {\tt AND}-gates in $\overline{S}$ already having in-neighbors that are not in $\mathcal{S}_t$;
\end{itemize}


Furthermore, the optimal solution of the main problem can be found either at $c[r,\{v_{out}\},\{v_{out}\},\{\}]$, if $f(v_{out}) = {\tt OR}$, or at $c[r,\{v_{out}\},\{\},\{\}]$, if $f(v_{out}) = {\tt AND}$, where $r$ is the root of the tree decomposition ${\mathcal T}$.
Recall that, for any node $t$ we assume that $v_{out} \in S$. 


In order to solve {\sc MinEC}$^{+}_{M}$, the counting of gates that output {\tt true} (in the solution) are made in introduce vertex nodes. 
Note that in the introduce node of an {\tt OR}-vertex $v$, it can not be simultaneously in $S$ and $\mathcal{B}^{\tt OR}$ 
because when a vertex is introduced then it is isolated in $C_{t}$ (we are considering an extended nice tree decomposition, i.e. the edges are introduced in introduce edge nodes). 

\paragraph*{Leaf node}

Let $t$ be a leaf node, then $X_t = \{v_{out}\}$. Since $v_{out}$ must be in $S$, then $\mathcal{B}^{\tt AND}=\emptyset$. 
Thus, we have three subproblems in Equation~(\ref{eq:leaf}).

\begin{equation}\label{eq:leaf}
c[t,\{v_{out}\},\mathcal{B}^{\tt OR},\mathcal{B}^{\tt AND}] = \left\{
\begin{array}{ll}
1, & \textrm{if  $f(v_{out})$ = {\tt AND}}\\
1, & \textrm{if } f(v_{out}) = {\tt OR} \textrm{ and }v_{out} \notin \mathcal{B}^{\tt OR}\\
\infty , & \textrm{if } f(v_{out}) = {\tt OR} \textrm{ and } v_{out} \in \mathcal{B}^{\tt OR}
\end{array}
\right.
\end{equation}

\paragraph*{Introduce vertex node}

Let $t$ be an introduce vertex node with exactly one child $t'$ such that $X_t = X_{t'} \cup \{v\}$. 
In the graph $C_{t}$, $v$ is an isolated vertex; consequently, as in the leaf nodes, there is infeasibility whenever $v$ belongs to $\mathcal{B}^{\tt OR}$ or $\mathcal{B}^{\tt AND}$. 
Besides, we have the possibility of $v$ be an input vertex ($f(v) \notin \{{\tt AND},{\tt OR}\}$) or  $v \notin S$, such situations only rescue previous subproblems without increment the current subsolution. On the other hand, we increment the subsolution by $1$ whenever $v \in S$. All possibilities are covered in Equations~(\ref{eq:IntroVertex01}),~(\ref{eq:IntroVertex02})~and~(\ref{eq:IntroVertex03}).

\begin{itemize}
    \item If $f(v) \notin \{{\tt AND},{\tt OR}\}$ then
\end{itemize}

\begin{equation}\label{eq:IntroVertex01}
    c[t,S,\mathcal{B}^{\tt OR},\mathcal{B}^{\tt AND}] = c[t',S \setminus \{v\},\mathcal{B}^{\tt OR},\mathcal{B}^{\tt AND}]
\end{equation}

\begin{itemize}
    \item If $f(v) = {\tt OR}$ then
\end{itemize}

\begin{equation}\label{eq:IntroVertex02}
    c[t,S,\mathcal{B}^{\tt OR},\mathcal{B}^{\tt AND}] = \left\{
    \begin{array}{ll}
        c[t',S,\mathcal{B}^{\tt OR},\mathcal{B}^{\tt AND}], & \textrm{if } v \notin S\\
        c[t',S \setminus \{v\},\mathcal{B}^{\tt OR},\mathcal{B}^{\tt AND}] + 1, & \textrm{if } v \in S \textrm{ and } v \notin \mathcal{B}^{\tt OR}\\
        \infty , & \textrm{if } v \in S \textrm{ and } v \in \mathcal{B}^{\tt OR}\\
    \end{array}
\right.
\end{equation}

\begin{itemize}
    \item If $f(v) = {\tt AND}$ then
\end{itemize}

\begin{equation}\label{eq:IntroVertex03}
    c[t,S,\mathcal{B}^{\tt OR},\mathcal{B}^{\tt AND}] = \left\{
    \begin{array}{ll}
        c[t',S \setminus \{v\},\mathcal{B}^{\tt OR},\mathcal{B}^{\tt AND}] + 1, & \textrm{if } v \in S\\
        c[t',S,\mathcal{B}^{\tt OR},\mathcal{B}^{\tt AND}], & \textrm{if } v \notin S \textrm{ and } v \notin \mathcal{B}^{\tt AND}\\
        \infty , & \textrm{if } v \notin S \textrm{ and } v \in \mathcal{B}^{\tt AND}\\
    \end{array}
\right.
\end{equation}

\paragraph*{Introduce edge node}

Let $t$ be an introduce edge node and $t'$ its child such that $X_t = X_{t'}$, which introduces the directed edge $uv$ such that $\{u,v\} \subseteq X_t$. Now, by including an edge, we can evaluate each subproblem concerning the sets $\mathcal{B}^{\tt OR}$ and $\mathcal{B}^{\tt AND}$; so, for each {\tt OR}-gate $ v \in S$, at least one in-neighbor also must be in $S$; 
so, either $uv$ attend this demand or another already introduced edge satisfied that.
We apply the same reasoning for {\tt AND}-gates: considering an {\tt AND}-gate $v \in \overline{S}$, then at least one in-edge of $v$ need comes to another vertex in $\overline{S}$; if $uv$ do not attend this requirement, the current subproblem is assigned to a previous subproblem where $v \in \mathcal{B}^{\tt AND}$.
All these conditions are handled in Equations~(\ref{eq:IntroEdge01})~and~(\ref{eq:IntroEdge02}). 
%
%
%
%
 Recall that we are introducing the directed edge $uv$.

\begin{itemize}   
\item If $f(v) = {\tt OR}$ then $c[t,S,\mathcal{B}^{\tt OR},\mathcal{B}^{\tt AND}]$ is equal to 
\end{itemize}
\begin{equation}\label{eq:IntroEdge01}
    \left\{
    \begin{array}{ll}
     c[t',S,\mathcal{B}^{\tt OR},\mathcal{B}^{\tt AND}], & \textrm{if } u \notin S\\
     \infty, & \textrm{if } u \in S \textrm{ and } v \notin S\cap \mathcal{B}^{\tt OR}\\
    \min \left\{
    c[t',S,\mathcal{B}^{\tt OR},\mathcal{B}^{\tt AND}], c[t',S,\mathcal{B}^{\tt OR} \setminus \{v\},\mathcal{B}^{\tt AND}] \right\}, & \textrm{if } u \in S \textrm{ and } v \in S\cap \mathcal{B}^{\tt OR}
    \end{array}
\right.
\end{equation}

\begin{itemize}
\item If $f(v) = {\tt AND}$ then $c[t,S,\mathcal{B}^{\tt OR},\mathcal{B}^{\tt AND}]$ is equal to
\end{itemize}
\begin{equation}\label{eq:IntroEdge02}
    \left\{
    \begin{array}{ll}
     c[t',S,\mathcal{B}^{\tt OR},\mathcal{B}^{\tt AND}], & \textrm{if } u \in S\\
     \infty , & \textrm{if } u \notin S \textrm{ and } v \in S\\
     \infty , & \textrm{if } \{u,v\} \subseteq \overline{S} \textrm{ and } v \notin \mathcal{B}^{\tt AND}\\
    \min \left\{ c[t',S,\mathcal{B}^{\tt OR},\mathcal{B}^{\tt AND}], c[t',S,\mathcal{B}^{\tt OR} ,\mathcal{B}^{\tt AND} \setminus \{v\}] \right\}, & \textrm{if } \{u,v\} \subseteq \overline{S} \textrm{ and } v \in \mathcal{B}^{\tt AND}
    \end{array}
\right.
\end{equation}

\paragraph*{Forget node}

Let $t$ be a forget node and $t'$ be its child such that $X_t = X_{t'} \backslash {v}$. In this case, we verify the best among either selecting or not $v$ in current subproblem. If $v$ is an input vertex, then this verification is trivial (it is enough to rescue the minimum subsolution varying only the membership of $v$ in $S$). For {\tt OR}-gates and {\tt AND}-gates, the same verification are made but considering the feasibility of $v$ through its membership in $\mathcal{B}^{\tt OR}$ and $\mathcal{B}^{\tt AND}$. Equations~(\ref{eq:forget01}),~(\ref{eq:forget02})~and~(\ref{eq:forget03}) summarize these three scenarios.

\begin{itemize}
\item If $f(v) \neq \{{\tt AND,OR}\}$ then
\end{itemize}
\begin{equation}\label{eq:forget01}
c[t,S,\mathcal{B}^{\tt OR},\mathcal{B}^{\tt AND}] = 
\min \left\{ c[t',S,\mathcal{B}^{\tt OR},\mathcal{B}^{\tt AND}], c[t',S \cup \{v\},\mathcal{B}^{\tt OR},\mathcal{B}^{\tt AND}] \right\}\\
\end{equation}


\begin{itemize}
\item If $f(v) = {\tt OR}$ then
\end{itemize}
\begin{equation}\label{eq:forget02}
c[t,S,\mathcal{B}^{\tt OR},\mathcal{B}^{\tt AND}] = 
\min \left\{ c[t',S,\mathcal{B}^{\tt OR},\mathcal{B}^{\tt AND}], c[t',S \cup \{v\},\mathcal{B}^{\tt OR} \cup \{v\},\mathcal{B}^{\tt AND}] \right\}\\
\end{equation}


\begin{itemize}
\item If $f(v) = {\tt AND}$ then
\end{itemize}
\begin{equation}\label{eq:forget03}
c[t,S,\mathcal{B}^{\tt OR},\mathcal{B}^{\tt AND}] = 
\min \left\{ c[t',S,\mathcal{B}^{\tt OR},\mathcal{B}^{\tt AND}  \cup \{v\}], c[t',S \cup \{v\},\mathcal{B}^{\tt OR},\mathcal{B}^{\tt AND}] \right\}\\
\end{equation}


\paragraph*{Join node}

Let $t$ be a join node with two children $t_1$ and $t_2$. For tabulation of the join nodes, we need to combine two partial solutions -- one originating from $C_{{t_1}}$ and another from $C_{{t_2}}$ -- in such a way that the merging is a feasible solution. 
Recall that $G$ is acyclic so we don't need to care about cycles.
Also, if a gate is activated in $C_t$ it must be activated in both children, so we must subtract duplicity. However, since each edge of $C_{t}$ is in either $C_{{t_1}}$ or ${C_{t_2}}$, the feasibility of merging children's solutions is guaranteed assuming that whether $v\in \mathcal{B}^{\tt OR}/\mathcal{B}^{\tt AND}$ then it is also in the respective set of one of the children, as described in Equation~\ref{eq:join}.
%
%

%

\begin{equation}\label{eq:join}
c[t,S,\mathcal{B}^{\tt OR},\mathcal{B}^{\tt AND}] = 
\min_{\mathcal{B}^{\tt OR}_1,\mathcal{B}^{\tt AND}_1,\mathcal{B}^{\tt OR}_2,\mathcal{B}^{\tt AND}_2}
\left\{ c[t_1,S,\mathcal{B}^{\tt OR}_1,\mathcal{B}^{\tt AND}_1] + c[t_2,S,\mathcal{B}^{\tt OR}_2,\mathcal{B}^{\tt AND}_2] 
\right\} - |S \setminus I| \\
\end{equation}

where $\mathcal{B}^{\tt OR}=\mathcal{B}^{\tt OR}_1\cup \mathcal{B}^{\tt OR}_2$ and $\mathcal{B}^{\tt AND}=\mathcal{B}^{\tt AND}_1\cup \mathcal{B}^{\tt AND}_2$.


Every bag of $\mathcal{T}$ has at most $tw + 2$ vertices (including $v_{out}$) and $v_{out}$ is fixed in the solution, thus each bag has at most $2^{tw+1}$ possible subsets $S$, there are at most $2^{tw+2}$ possible sets $\mathcal{B}^{\tt OR}$, and there are at most $2^{tw+1}$ sets $\mathcal{B}^{\tt AND}$. Therefore, the entire matrix has size $2^{O(tw)}\cdot n$. As each entry of the table can be computed in $2^{O(tw)}$ time, it holds that the algorithm performs in time $2^{O(tw)}\cdot n$.
\end{proof}

From Theorem~\ref{treewidthGenus} and Theorem~\ref{theo:dp}, it follows that Corollary~\ref{corollary:genus} holds.

\begin{corollary}
\label{corollary:genus}
{\sc MinEC$^+_M$} can be solved in time $2^{O(k\cdot (g+1)^{3/2})} \cdot n^{O(1)}$, where $g$ is the genus of the input.
\end{corollary}

Finally, since any satisfying assignment $X$ produces a partition of the vertices of $C$ according to the properties described in Theorem~\ref{theo:dp}, to find $X$ which minimizes $|\mathcal{S}_X \setminus I|$ is LinEMSOL$_1$-expressible (see~\cite{courcelle2012graph} for LinEMSOL details). Thus, the following holds.

\begin{theorem}
{\sc MinEC$^+_M$} is fixed-parameter tractable when parameterized by the clique-width of the input.
\end{theorem}

\section{Conclusions}

The energy complexity measure represents an interesting manner to an\-a\-lyse the activation of gates through a circuit. Previous works address energy complexity in threshold circuits as a model that simulates a neural network. By analyzing such a measure restricted to Boolean circuits, recent bounds in circuit complexity analysis were presented in \cite{dinesh2020new} and~\cite{sun2021relationship}. 

In this paper, we introduce the discussion of energy complexity problems in terms of time complexity. We investigate the time complexity to computing the best-case energy complexity among satisfying assignments of monotone circuits (denoted by {\sc MinEC}$^{+}_{M}$) and its respective version parameterized by the size of the solution ($k$-{\sc MinEC}$^{+}_{M}$). 
We prove that {\sc MinEC}$^{+}_{M}$ is NP-complete on planar graphs, and that $k$-{\sc MinEC}$^{+}_{M}$ is W[1]-hard, but in XP.
In addition, we show that {\sc MinEC}$^{+}_{M}$ is fixed-parameter tractable when parameterized either by the treewidth of the input circuit or by the genus of the input plus the size of the solution.
Besides that, we also remark that $k$-{\sc MinEC}$^{+}_{M}$ is FPT when restricted to instances with bounded in-degree, and that {\sc MinEC}$^{+}_{M}$ is FPT concerning the clique-width of the directed graph/circuit $C$.

\medskip

Since $k$-{\sc MinEC}$^{+}_{M}$ is in XP, we left open whether {\sc $k$-MinEC}$^{+}_{M}$ is in W[P]. Also, the parameterized complexity of $k$-{\sc MinEC}$^{+}_{M}$ on circuits that are planar but have {\tt NOT} gates seems also interesting. In addition, the (parameterized) complexity of computing the worst-case energy complexity of circuits that are not monotone is also interesting. 

Finally, we remark that there are several interesting energy complexity problems that can arise varying from the options ``worst-case or best-case'',  ``general assignments or satisfying assignments'', ``monotone or non-monotone'', and  ``planar or non-planar''.


\section*{Acknowledgement}


This research has received funding from Rio de Janeiro Research Support Foundation (FAPERJ) under grant agreement E-26/201.344/2021,  National Council for Scientific  and Technological Development (CNPq) under grant agreement 309832/2020-9, 
and the European Research Council (ERC) under the European Union's Horizon $2020$ research and innovation programme under grant 
\begin{minipage}{0.8\textwidth}
agreement CUTACOMBS (No. $714704$).
\end{minipage}
\begin{minipage}{0.2\textwidth}
        \includegraphics[scale=.3]{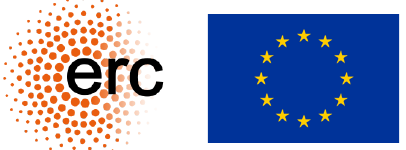}
\end{minipage}

\bibliographystyle{plainurl} 
\bibliography{mybibfile}

\end{document}